\documentclass[peerreview]{IEEEtran}
\usepackage{graphicx, amssymb, latexsym, amsfonts, amsmath, lscape,  amscd, multicol, multirow, diagbox, amsthm, color, epsfig, mathrsfs, tikz, enumerate, url, hyperref}

\usepackage[noadjust]{cite}

\usepackage{realboxes}
\usepackage{diagbox}

\usepackage{todonotes}

\newtheorem{theorem}{Theorem}

\newtheorem{corollary}[theorem]{Corollary}

\newtheorem{lemma}[theorem]{Lemma}
\newtheorem{proposition}[theorem]{Proposition}

\theoremstyle{definition}
\newtheorem{defn}[theorem]{Definition}
\newtheorem{example}[theorem]{Example}
\newtheorem{observation}[theorem]{Observation}

\newtheorem{remark}[theorem]{Remark}

\renewcommand\vec{\mathbf}

\begin{document}

\title{Unique Insertion Error
Patterns in Levenshtein’s Reconstruction Problem\thanks{The authors were funded in part by the Research Council of Finland grants 338797 and 358718. A conference version of this paper appeared in 2025 IEEE Information Theory Workshop (ITW2025) \cite{pavanITWInsertion}.}
}

\author{\textbf{Ville Junnila}, \textbf{Tero Laihonen}, \textbf{Tuomo Lehtil{\"a}} and \textbf{Pavan Padavu Devaraj}\\
Department of Mathematics and Statistics\\
University of Turku, FI-20014 Turku, Finland\\
 Email: \{viljun, terolai, tualeh\}@utu.fi, pavanpdevaraj@gmail.com
}

\date{}

\maketitle

\begin{abstract}
    
		Levenshtein's sequence reconstruction model plays an essential role in information retrieval of advanced memory systems, such as the DNA-based storage systems. In the Levenshtein's model, a word $\vec{x}\in \mathbb{Z}_q^n$ is transmitted through $N$ noisy channels, and the goal is to recover, using the output words produced by these channels, the original word $\vec{x}$ unambiguously, or with small uncertainty $\mathcal{L}$. Errors occurring in the channels usually involve substitutions, insertions and deletions. In this work, we focus on insertion errors.   
        One of the main questions in this context is determining the minimum number of channels $N$ required to recover the transmitted word $\vec{x}$. The original formulation of Levenshtein's sequence reconstruction problem requires that all the output words from the channels are distinct. However, channels may produce the same output word even if different insertion errors occur in them. In this paper, we investigate two reconstruction models  where the channels are allowed to produce identical output words even though different insertion errors occur in the channels. These two models, called \emph{the multiset model} and \emph{non-multiset model}, generalize the Levenshtein's model. Let us denote the minimum number of channels  required to \emph{unambiguously} recover the transmitted word $\vec{x}\in \mathbb{Z}_q^n$ by $N_q^m(n,t)+1$ in the multiset model and  $N_q^{nm}(n,t)+1$ in the non-multiset model, where $t$ denotes the exact number of insertions occurring in a channel. We determine $N_q^{m}(n,1)$ and $N_q^{nm}(n,1)$ for all $n$ and $q$, and show the somewhat surprising fact that $N_q^m(n,1)=N_q^{nm}(n,1)$. Moreover, we provide a full characterization of the words 
        that attain this value. 
        We also give a general lower bound on $N_q^m(n,t)$ for $t\ge 1$ and a recursive upper bound.  
        For $t=1$, we consider a construction from codes $C\subseteq \mathbb{Z}_q^n$ to codes $C'\subseteq \mathbb{Z}_q^{n+2}$ such that the number of channels required to determine the transmitted word $\vec{x}\in C'$ is small. This construction is shown to be optimal for certain parameters.

	\end{abstract}

	\section{Introduction}
    
	\emph{Levenshtein's sequence reconstruction problem}, introduced in \cite{Levenshtein}, has gained renewed attention due to its relevance in information retrieval for advanced storage technologies, such as DNA based ones~\cite{horovitz2018reconstruction, yaakobi2016constructions}. In the information retrieval process of DNA data storage (see \cite{bornholt2016dna, church2012next, grass2015robust, yazdi2015dna}), numerous copies of the stored information are obtained, each typically affected by substitution, deletion, and insertion errors. 
    The goal is to recover the original information using these erroneous copies. For  results on this problem, see, for example, \cite{Levenshtein, horovitz2018reconstruction, levenshtein2005reconstruction, gabrys2018sequence,  Maria_Abu-Sini, Uusi_Maria_Abu-Sini, Abu-Sini2024, PHAM2025105980, sabary2020, srinivasavaradhan2018, srinivasavaradhan2019, barlev2021}.  
	
	 Let us first introduce some notation. We represent the set $\{a, a + 1, \dots, b\}$ by $[a, b]$ for integers $a \le b$. Let $\mathbb{Z}_q= \{0, 1, \dots, q - 1\}$ denote the ring of $q \ge 2$ elements, and $\mathbb{Z}_q^n = \mathbb{Z}_q \times \dots \times \mathbb{Z}_q\ (n \text{ times})$. For a word $\vec{x} = x_1\dots x_n\in \mathbb{Z}_q^n$, we denote by $\vec{x}_{[a,b]}$ the shortened subword $x_a x_{a + 1}\dots x_b \in \mathbb{Z}_q^{b - a + 1}$ of $\vec{x}$. The all-zero word $00 \dots 0 \in \mathbb{Z}_q^n$ is denoted by $\vec{0}$ and the empty word by $\varepsilon$. The \emph{Hamming weight} $w(\vec{x})$ of $\vec{x} \in \mathbb{Z}_q^n$ is the number of non-zero coordinates of $\vec{x}$. A non-empty subset of $\mathbb{Z}_q^n$ is called a \emph{code} and its elements are called \emph{codewords}. For a  set $A$, the notation $|A|$ is the usual cardinality of the set.   
     Given a \emph{multiset} $A$, that is, a collection of elements in which elements can be repeated multiple times, let $\mathrm{set}(A)$ denote the set of distinct words in $A$ and let $m(\vec{a}, A)$ denote the multiplicity of $\vec{a}$ in the multiset $A$.
     If $A$ is a multiset, then by $|A|$ we denote the total cardinality of the multiset, that is, the sum of multiplicities in $A$ of the different elements of $\mathrm{set}(A)$. 
     
	Let $C \subseteq \mathbb{Z}_q^n$ be a code. In Levenshtein's sequence reconstruction problem, one transmits a word $\vec{x}\in C$ through $N$ channels. The channels introduce some errors which can be substitutions, deletions, and insertions to the word $\vec{x}$ and one obtains the set of output words $Y = \{\vec{y}_1, \vec{y}_2, \dots, \vec{y}_N\}$ from the channel (see Fig.~\ref{LevenshteinFig}). The output of each channel is assumed to be \emph{different} from each other, and the number of errors that may occur in a channel is bounded by a parameter $t$. Let $T(Y) \subseteq C$ be the set of codewords that can give the set of output words $Y$ when transmitted through the $N$ channels. Clearly, $\vec{x}\in T(Y)$. Given $C$, the maximum possible size of $T(Y)$ over all possible transmitted words $\vec{x}\in C$ is denoted by $\mathcal{L}$. If $\mathcal{L}=1$, the transmitted word can be uniquely determined based on $Y$. This formulation of the problem is called the \emph{Levenshtein's (traditional) model}.

    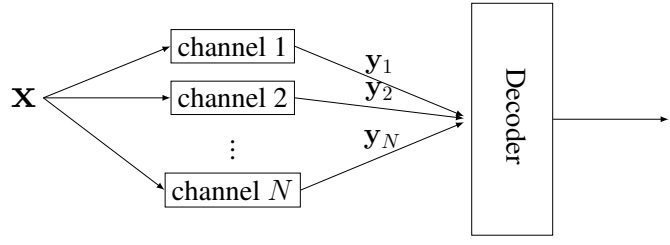
\begin{figure}
    \centering
    \resizebox {.5\textwidth} {!} {
    \begin{tikzpicture}[>=latex, node distance=2cm]

\node[draw, rectangle] (ch1) {\Large{channel 1}};
\node[draw, rectangle, below=3mm of ch1] (ch2) {\Large{channel 2}};
\node[draw=none, below=1mm of ch2] (dots) {\Large{$\vdots$}};
\node[draw=none, below=-0.5mm of ch2] (origin) {};
\node[draw, rectangle, below=1mm of dots] (chN) {\Large{channel $N$}};

\node[left=2.2cm of ch2] (x) {\huge{$\mathbf{x}$}};

\node[draw, rectangle, minimum height=4cm, minimum width=14mm, 
      right=4cm of origin] (decoder) {};
\node[rotate=270] at (4.9,-1.2) {\Large{Decoder}};

\node (merge) at (decoder.west) {};

\node[right=2cm of decoder] (out) {};

\draw[->] (x.east) -- (ch1.west);
\draw[->] (x.east) -- (ch2.west);
\draw[->] (x.east) -- (chN.west);

\draw[->] (ch1.east) -- node[above] {\Large{$\vec{y}_1$}} (merge);
\draw[->] (ch2.east) -- node[above] {\Large{$\vec{y}_2$}} (merge);
\draw[->] (chN.east) -- node[above] {\Large{$\vec{y}_N$}} (merge);

\draw[->] (decoder.east) -- node[above] {} (out);

\end{tikzpicture}}
    \caption{The Levenshtein's sequence reconstruction model.}
    \label{LevenshteinFig}
\end{figure}

    When at most $t$ \emph{substitution} errors occur in each channel, Levenshtein~\cite{Levenshtein} determined the minimum number of channels $N$ giving $\mathcal{L}=1$ over any $q$-ary code $C$. For the binary case ($q = 2$) and $\mathcal{L}\le 2$, see Yaakobi and Bruck~\cite{yaakobi2018uncertainty}. Further results for substitutions and larger values of $\mathcal{L}$ and $N$ can be found, for example, in~\cite{junnila2024levenshtein, junnila2021levenshtein,pavanITWSubstitution,pavanISITsubstitution,junnila2026sizeintersectionqaryhamming}.
     For \emph{insertion errors} in the traditional model, the following result for $\mathcal{L} = 1$ is due to Levenshtein, which is based on Equation~(51) and Theorem~3 of \cite{levenshtein2001efficient}.
    
    \begin{theorem}[\cite{levenshtein2001efficient}]\label{TheLevInsBound}
        Let exactly $t$ insertion errors occur in the traditional Levenshtein's channel model and $C=\mathbb{Z}_q^n$. Then $\mathcal{L}=1$ if and only if the number of different output words satisfies $N\ge N_q^+(n,t)+1$, where
    \begin{equation}\label{OrigLevIns}
    N_q^+(n,t)= \sum_{i=0}^{t-1}\binom{n+t}{i}(q-1)^i(1-(-1)^{t-i}).
    \end{equation}
\end{theorem}

    For further relationships between the parameters $N$ and $\mathcal{L}$ under insertion or deletion errors, see~\cite{Levenshtein, Uusi_Maria_Abu-Sini, Abu-Sini2024, PHAM2025105980, abbasian2024size}. For further work on insertion errors when $C \neq \mathbb{Z}_q^n$, see~\cite{sala2017exact} and \cite{ye2023reconstruction}.

The traditional Levenshtein's model requires that all the output words are \emph{distinct}. This, however, means that the information contained in the multiplicities of the output words is lost. Consider, for example, 
$\vec{x}=00\in \mathbb{Z}_2^2$ with exactly $t = 1$ insertion. Since in the traditional model the output words must be distinct, the received set of output words is a subset of $Y_1=\{000,100,010,001\}$. However, if we assume that in each channel the output word is obtained by a different insertion error, then the output words are from the \emph{multiset} $Y_2=\{\underline{0}00,0\underline{0}0,00\underline{0},\underline{1}00,0\underline{1}0,00\underline{1}\}$ where the different insertion errors are indicated by underlining. This means that the probability of receiving the word $000$ increases from 25\% in $Y_1$ to 50\% in $Y_2$. Moreover, it seems natural that $000$ is more likely to be the output word than the other words, since there are 3 different ways to obtain it while the other words only have 1. These observations motivate the generalized models introduced in Definitions~\ref{DefMultiModel} and \ref{DefNonMultiModel}, which take into account the possibility that the same output word can arise from different channels (as in the example above, where the word $000$ occurs with multiplicity $m(000,Y_2)=3$ although $\mathrm{set}(Y_2)=Y_1$).  
These models were introduced in \cite{junnila2023levenshtein,junnila2025levenshtein}, where the focus was on \emph{deletion} errors. In this paper, we consider these models for \emph{insertion} errors. It is pointed out in \cite{junnila2025levenshtein} that for the third type of error mentioned above, namely the \emph{substitution} errors, the new models coincide with the traditional Levenshtein's model, since different substitution errors in channels necessarily lead to different output words -- unlike in the deletion or insertion error case.

Let us then give the formal definitions of the models. 
An \emph{insertion vector} $w$ is a list of $n + 1$ ($q$-ary) words $(\vec{w_0}, \vec{w_1}, \dots, \vec{w_n})$ of total length at most $t$, also called the \emph{weight} of the vector, where $\vec{w_i}$ can be the empty word $\varepsilon$ (which is of zero length and weight). When an insertion error occurs on a word $\vec{x} \in \mathbb{Z}_q^n$, the word $\vec{w_{i}}$ is inserted after the $i$th symbol of $\vec{x}$ for each $1 \leq i \leq n$, and $\vec{w_0}$ is inserted before the first symbol of $\vec{x}$. For example, if $t = 2$, then the word $10001$ can be obtained from $\vec{x} = 100$ using, for instance, different insertion vectors $(\varepsilon,0,\varepsilon,1)$ or $(\varepsilon, \varepsilon,\varepsilon,01)$ of total length two (or weight two). Let us now define the \emph{multiset model} and \emph{non-multiset model} for insertion errors. Throughout the paper, we assume that $n\ge 1$, $q\ge 2$ and $t\ge 1$ unless otherwise stated.

\begin{defn}\label{DefMultiModel} Let  $C \subseteq \mathbb{Z}_q^n$ be such that $|C|\ge 2$. In the \emph{multiset model for insertion errors}, a codeword $\vec{x}\in C$ is transmitted through $N$ channels, each applying a unique insertion vector of weight at most $t$ to it. This gives us a multiset $Y^m$ of $N$ output words $\vec{y_i}$, where $1\leq i\leq N$.
We denote by $T(Y^m)$ the set of codewords such that if $\vec{x'}\in T(Y^m)$, then the output multiset $Y^m$ can be obtained when $\vec{x'}$ is transmitted through the $N$ channels. Given $C$, the maximum possible size of $T(Y^m)$, over all possible transmitted words $\vec{x}\in C$, is denoted by $\mathcal{L}$. As $\vec{x}\in T(Y^m)$, we have $\mathcal{L}\ge 1.$
\end{defn}

\begin{defn}\label{DefNonMultiModel} Let  $C \subseteq \mathbb{Z}_q^n$ be such that $|C|\ge 2$.
In the \emph{non-multiset  model for insertion errors}, a codeword $\vec{x}\in C$ is transmitted through $N$ channels, each applying a unique insertion vector of weight at most $t$ to it, giving a multiset $Y^m$.
We receive the pruned set $Y=\mathrm{set}(Y^m)$ and we know $N$.
Let $T(Y)$ be the set of codewords such that if $\vec{x'}\in T(Y)$, then the output set $Y$ can be obtained when $\vec{x'}$ is transmitted through the $N$ channels. Given $C$, the maximum possible size of $T(Y)$, over all possible transmitted words $\vec{x}\in C$, is denoted by $\mathcal{L}$. 
As $\vec{x}\in T(Y)$, we have $\mathcal{L}\ge 1.$
\end{defn}

The multiset model assumes that each channel introduces different insertion errors (but the output words are allowed to be identical). This is an idealized mathematical model, but the assumption seems to be reasonable for a wide range of parameters $n$, $t$, $q$ and $N$ of the underlying code $C$ (see Section II-F in \cite{junnila2025levenshtein}) when the different insertion vectors are assumed to be equally likely. When we compare the multiset model to the non-multiset model, note that the latter can cope with some channels having identical insertion errors (so not just output words being equal), provided that at least a certain number (the parameter $N$ in Definition~\ref{DefNonMultiModel})  of channels with distinct insertion errors exist. In other words, this allows the non-multiset model to cope with some common systematic errors across different channels. Notice that the probability of obtaining the required number $N$ of channels with different error patterns in the non-multiset case can be increased by increasing the total number of channels  (see \cite[Section II-F]{junnila2025levenshtein}). An important observation is that we \emph{do not} need to know in the non-multiset model which channels have different error patterns and which have the same patterns.  Regarding the information obtained from the channels, the difference between the two models can be understood in the following way: in the multiset model we receive $Y^m$, that is, $\mathrm{set}(Y^m)$ together with the multiplicity of every element in the multiset, and in the non-multiset model we receive $\mathrm{set}(Y^m)$ knowing also that there are (at least) $N=|Y^m|$ channels with different insertion errors. 

Notice that the models are different from the trace reconstruction problem~\cite{viswanathan2008improved} (see also \cite{bhardwaj2020trace} for a survey on trace reconstruction from the perspective of computational biology). Indeed, the trace reconstruction problem is probabilistic in nature, whereas our models are combinatorial problems usually called \emph{adversarial problems}. 
The number of insertions in a word is bounded by $t$ in our models, whereas typically in the trace reconstruction problem each insertion depends on a given probability $p$ and is independent from other insertions.

In this paper, we want to recover the transmitted word unambiguously, that is, to have $\mathcal{L}=1$. The main question regarding the models above is to find (similarly to \eqref{OrigLevIns} in the traditional model) the smallest possible $N$ satisfying $\mathcal{L}=1$. Such parameters $N$ for the models of Definition~\ref{DefMultiModel} and~\ref{DefNonMultiModel} actually do exist (as will be pointed out in Remark~\ref{Nexists}). We denote the smallest possible $N$ in the multiset model by $N_q^m(C;n,t)+1$ and in the non-multiset model by $N_q^{nm}(C;n,t)+1$. Furthermore, we begin our paper by considering (as in Theorem~\ref{TheLevInsBound}, see also~\cite{Maria_Abu-Sini,Abu-Sini2024}) the \emph{uncoded} case, that is, $C=\mathbb{Z}_q^n$. In that case, we drop $C$ from the notation, i.e., $N_q^m(\mathbb{Z}_q^n;n,t)=N_q^m(n,t)$ and $N_q^{nm}(\mathbb{Z}_q^n;n,t)=N_q^{nm}(n,t).$

Because we aim for $\mathcal{L}=1$, we wish to find the smallest number of channels $N$ such that if we transmit two different words, say $\vec{x}$ and $\vec{x}'$, then the obtained multisets of output words (resp. the related sets) are never the same in the multiset model (resp. in the non-multiset case).
To this end, in Section~\ref{sec:preliminaries}, we present the basic properties of the above-described models which will help us in obtaining and understanding the results in this paper. 
Further in Section~\ref{BoundAndExtremalPairs}, we study the number of channels required by both models 
when one insertion takes place, that is, $t=1$. In this process, we develop methods that allow us also to determine all extremal word pairs, that is, word pairs that require the maximum possible number of channels to distinguish. The methods obtained in this analysis are also useful in Section~\ref{sec:N_q^m(C;n,1)} where we provide a construction from a code $C\subseteq \mathbb{Z}_q^n$ with small $N_q^m(C;n,1)$ to a new code $C'\subseteq \mathbb{Z}_q^{n+2}$ with small $N_q^m(C';n+2,1)$. This construction is optimal for certain parameters. 
In Section~\ref{sec:generalbds}, we derive some general bounds for $t\ge 1$ 
regarding the number of channels required 
to distinguish words in the multiset model.
Finally, we conclude our work in Section~\ref{sec:conc}, by summarizing and also presenting some directions for possible future work.

\section{Basics}
\label{sec:preliminaries}
We begin by introducing some notation. For two multisets $A_1$ and $A_2$, we define the following simple operations:
\begin{enumerate}
    \item[M1:] $A_1 \oplus A_2$ is a multiset whose distinct elements are $\mathrm{set}(A_1) \cup \mathrm{set}(A_2)$ and for any element $\vec{a} \in A_1 \oplus A_2$, $m(\vec{a}, A_1 \oplus A_2) = m(\vec{a}, A_1) + m(\vec{a}, A_2)$. If $A_2 = \{a\}$, then we represent $A_1 \oplus A_2$ as $A_1 \oplus a$, (that is, omitting the brackets) for convenience (also  $A_2 \oplus A_1$ can be written as $a \oplus A_1$).
    \item[M2:] $A_1 \ominus A_2$ is a multiset whose elements belong to $\mathrm{set}(A_1)$ and for any element $\vec{a} \in A_1 \ominus A_2$, $m(\vec{a}, A_1 \ominus A_2) = \max\{m(\vec{a}, A_1) - m(\vec{a}, A_2), 0\}$.
    \item[M3:] $A_1 \cap A_2$ is a multiset whose elements are in $\mathrm{set}(A_1) \cap \mathrm{set}(A_2)$ and for any element $\vec{a} \in \mathrm{set}(A_1) \cap \mathrm{set}(A_2)$, $m(\vec{a}, A_1 \cap A_2) = \min(m(\vec{a}, A_1), m(\vec{a}, A_2))$.
\end{enumerate}

\smallskip
In this work, we follow the convention that when the objects under discussion are sets, then $\cap$ refers to the usual set intersection, and if the objects under discussion are multisets, then $\cap$ refers to the multiset intersection of M3.

Recall that our focus is on determining the transmitted word uniquely when $C = \mathbb{Z}_q^n$. In particular, our goal is to determine the minimum number $N$ of channels required to always distinguish two words $\vec{x}, \vec{x'} \in \mathbb{Z}_q^n$.

Observe that the multiset of words obtained by performing \emph{at most} $t$ insertions on a word $\vec{x}$ is equivalent to the disjoint union of the multisets of words obtained by performing exactly $i$ insertions for $0 \leq i \leq t$. Hence, it is customary (see Theorem~\ref{TheLevInsBound}) to focus on the case with \emph{exactly} $t$ insertions on $\vec{x}$ and our work also follows this convention.

Let $M_t^{n, q}(\vec{x})$ denote the \emph{$t$-insertion multiset sphere} centered at $\vec{x}$, that is, the multiset of words obtainable by performing exactly $t$ insertions on $\vec{x}$. This notation is abbreviated as $M_t(\vec{x})$ when the context is clear. It is shown in~\cite{junnila2025levenshtein} that the size of $M_t^{n, q}(\vec{x})$, denoted by $M_q(n, t)$, is independent of the word $\vec{x}$ and is given by the formula
    \begin{equation}
        \label{eq:multisetspheresize}
        |M_t^{n, q}(\vec{x})| = M_q(n, t) = q^t \binom{n + t}{t}.
    \end{equation}
Observe that $M_q(n, t)$ also gives the number of insertion vectors made up of $n + 1$ $q$-ary words of total length $t$. If we  denote, as in \cite{levenshtein2001efficient}, by $I_t(\vec{x})$  the set of all distinct words obtained from $\vec{x}$ by $t$ insertions, then  
$\mathrm{set}(M_t(\vec{x}))=I_t(\vec{x}).$ Moreover, by \cite{levenshtein2001efficient}, we have $N_q^+(n,t)=\max_{\vec{x},\vec{x'}\in\mathbb{Z}_q^n,\vec{x}\neq \vec{x'}}|I_t(\vec{x})\cap I_t(\vec{x'})|.$

Next we consider the two models in more detail and illustrate some concepts with an example.

\subsection{On the Multiset Model}
Consider first the multiset model. Let us denote the intersection 
of the multiset spheres centered at $\vec{x}$ and $\vec{x'}$ by $A = M_t(\vec{x}) \cap M_t(\vec{x'})$. The different words in $A$ form  $\mathrm{set}(A)$.
Let us denote by $N^m_{\vec{x},\vec{x'}}$ the size of the  intersection.
Hence,
\begin{equation}\label{Nmxx}
N^m_{\vec{x},\vec{x'}} =|M_t(\vec{x}) \cap M_t(\vec{x'})|= \sum_{i = 1}^{|\mathrm{set}(A)|}\min (m(\vec{a_i}, M_t(\vec{x})), m(\vec{a_i}, M_t(\vec{x'}))).
\end{equation}
Thus, the number of channels required (in the worst case)  to distinguish between $\vec{x}$ and $\vec{x'}$ is $N^m_{\vec{x},\vec{x'}} + 1$ (for an illustration, see Example~\ref{NxxExample}). Of course, we need to distinguish every pair of words. Recalling the definition of $N^m_q(n, t)$ from the introduction, it can also be interpreted as the maximum size of the intersection of $t$-insertion multiset spheres centered at two $q$-ary words of length $n$. By \eqref{Nmxx}, we get
\begin{equation}\label{Nm}
N^m_q(n, t) = \max_{\vec{x},\vec{x'} \in \mathbb{Z}_q^n, \vec{x}\neq \vec{x'}} N^m_{\vec{x},\vec{x'}}.
\end{equation}
Hence, the number of channels required to distinguish all pairs of words in the multiset model is given by $N^m_q(n, t) + 1.$ This number of different error patterns in the multiset model ensures that there exists an instance of an output word (with distinguishing multiplicity) such that it is obtained either from $\vec{x}$, or from $\vec{x'}$, but not from both.

\subsection{On the Non-Multiset Model}
Let us now consider the non-multiset case. We introduce the following notation for convenience where, again, $A = M_t(\vec{x}) \cap M_t(\vec{x'})$. For two distinct words $\vec{x}$ and $\vec{x'}$, we denote $$S_{t, \vec{x'}}(\vec{x}) = \sum_{\vec{a} \in \mathrm{set}(A)} m(\vec{a}, M_t(\vec{x})),$$ that is, $S_{t, \vec{x'}}(\vec{x})$ counts the number of insertion vectors of weight $t$ which, when applied to $\vec{x}$, give an output word which can also be obtained from the word $\vec{x'}$ by some insertion vector (see Example~\ref{NxxExample} for an illustration) of weight $t$. Moreover, if $\vec{x}$ is transmitted and more than $S_{t, \vec{x'}}(\vec{x})$ outputs are received with distinct insertion vectors, then we necessarily receive a word, which does not belong to $M_t(\vec{x'})$.

To guarantee that we can always distinguish $\vec{x}$ and $\vec{x'}$ in the non-multiset model, we need to receive from at least one channel an output word that does not belong to the set of common words $\mathrm{set}(A)$. Let us now consider how many channels with different insertion vectors we need to be sure of that.
The number of distinct error patterns that generate the output words in $\mathrm{set}(A)$ from $\vec{x}$ and from $\vec{x'}$ is given by $S_{t, \vec{x'}}(\vec{x})$ and by $S_{t, \vec{x}}(\vec{x'})$, respectively. Let us denote
\begin{equation}\label{Nnmxx}
    N^{nm}_{\vec{x},\vec{x'}} = \min\left(S_{t, \vec{x'}}(\vec{x}), S_{t, \vec{x}}(\vec{x'})\right).
\end{equation}
Thus, if we have at most $N^{nm}_{\vec{x},\vec{x'}}$ channels with distinct insertion errors, then we cannot always distinguish between transmitted words $\vec{x}$ and $\vec{x'}$, but 
$N_{\vec{x}, \vec{x'}}^{nm} + 1$ such channels is enough to know the transmitted word (see Example~\ref{NxxExample} for an illustration). Recalling the definition of $N^{nm}_q(n, t)$ from the introduction, it can also be interpreted as
\begin{equation}\label{Nnm}
    N^{nm}_q(n, t) = \max_{\vec{x},\vec{x'} \in \mathbb{Z}_q^n, \vec{x}\neq \vec{x'}}N^{nm}_{\vec{x},\vec{x'}}
\end{equation}
by comparing all possible pairs of distinct words. Hence, the number of channels required to distinguish every pair of words in the non-multiset model is given by $N^{nm}_q(n, t) + 1.$ In the following example, we illustrate the previous definitions.

\begin{example}\label{NxxExample}
Consider $\vec{x} = 01$ and $\vec{x'} = 11$, and suppose that exactly $t = 2$ insertion errors occur in a channel. When $q = 2$, the output words which can be obtained from $\vec{x}$ and from $\vec{x'}$, by all possible insertion vectors, are listed in Table~\ref{tab:multnonex}. 
\begin{table}[h]
    \caption{The multiset spheres $M_2(01)$ and $M_2(11)$. The underlined symbols are the inserted ones.}
    \label{tab:multnonex}
    \centering
    \begin{tabular}{c|c|c|c|c|c|c|c|}
         \hline
         \multirow{6}{*}{$M_2(01)$} & \underline{00}01 & \underline{0}01\underline{0} & \underline{0}0\underline{1}1 & 01\underline{00} & \underline{01}01 & 0\underline{1}1\underline{0} & 0\underline{11}1 \\
                                    & \underline{0}0\underline{0}1 & 0\underline{0}1\underline{0} & \underline{0}01\underline{1} &      & 0\underline{10}1 & 01\underline{10} & 0\underline{1}1\underline{1} \\
                                    & 0\underline{00}1 &      & 0\underline{01}1 &      & 01\underline{01} &      & 01\underline{11} \\
                                    &      &      & 0\underline{0}1\underline{1} &      &      &      &      \\
         \cline{2-8}
          & \underline{10}01 & \underline{1}01\underline{0} & \underline{1}01\underline{1} & \underline{11}01 \\
          & \underline{1}0\underline{0}1 &      & \underline{1}0\underline{1}1 &      \\
          \hline
          \hline
          \multirow{9}{*}{$M_2(11)$} & \underline{00}11 & \underline{0}1\underline{0}1 & \underline{0}11\underline{0} & \underline{01}11 & 1\underline{00}1 & 1\underline{0}1\underline{0} & \underline{10}11 \\
                                    &      &      &      & \underline{0}1\underline{1}1 &      &      & 1\underline{01}1 \\
                                    &      &      &      & \underline{0}11\underline{1} &      &      & 1\underline{0}1\underline{1} \\
         \cline{2-8}
          & 11\underline{00} & \underline{1}1\underline{0}1 & \underline{1}11\underline{0} & \underline{11}11\\
          &      & 1\underline{10}1 & 1\underline{1}1\underline{0} & \underline{1}1\underline{1}1\\
          &      & 11\underline{01} & 11\underline{10} & \underline{1}11\underline{1}\\
          &      &      &      & 1\underline{11}1 \\
          &      &      &      & 1\underline{1}1\underline{1} \\
          &      &      &      & 11\underline{11} \\
          
          \hline
          \hline
    \end{tabular}
\end{table}

Observe that in this case, $\mathrm{set}(A)=\mathrm{set}(M_2(01) \cap M_2(11)) = \{0011, 0101, 0110, 0111, 1001, 1010, 1011, 1101\}$. 
By counting the multiplicities of each of the words in $\mathrm{set}(A)$ (in the given order) in $M_2(01)$ and $M_2(11)$, we get, using \eqref{Nmxx}, that $N^m_{01, 11} =\min(4, 1) + \min(3, 1) + \min(2, 1) + \min(3, 3) + \min(2, 1) + \min(1, 1) + \min(2, 3) + \min(1, 3) = 11$ and, using \eqref{Nnmxx}, that $N^{nm}_{01, 11} =  \min\left(S_{2, 11}(01), S_{2, 01}(11)\right) = \min(4 + 3 + 2 + 3 + 2 + 1 + 2 + 1, 1 + 1 + 1 + 3 + 1 + 1 + 3 + 3) = \min(18, 14) = 14$.
Consequently, it follows that $11 + 1 = 12$ (resp. $14 + 1 = 15$) channels with different insertion errors guarantee in the multiset model (resp. in the non-multiset model) that we can always distinguish between the transmitted words $\vec{x}=01$ and $\vec{x'}=11$, and for fewer channels this does not hold.
\end{example}

In the following remark, we show that the values $N_q^m(C;n,t)$ and $N_q^{nm}(C;n,t)$ always exist for every choice of parameters $n\ge 1$, $q\ge 2$ and $t\ge 1$.
\begin{remark}\label{Nexists}
    In what follows, we point out that the parameter $N$ giving $\mathcal{L}=1$ in Definition~\ref{DefMultiModel} always exists (that is, it is a finite integer), and the same is true for the corresponding parameters in Definition~\ref{DefNonMultiModel}. In other words, there are parameters $N_q^m(C;n,t)+1$ and $N_q^{nm}(C;n,t)+1$ such that we can always determine the transmitted word unambiguously. Let us next verify this with the aid of the following three steps. 
    
    1) First, we show that from the existence of $N_q^{nm}(C;n,t)$, the existence of $N_q^m(C;n,t)$ follows.
    Indeed, if $N_q^{nm}(C;n,t)$ exists, then
\begin{equation}\label{mLessnm}
    N_q^m(C;n,t)\le N_q^{nm}(C;n,t).
\end{equation}
 We show next that \eqref{mLessnm} holds. If for two multisets $Y_1^m$ and $Y_2^m$ it is true that the corresponding pruned sets (as in Definition~\ref{DefNonMultiModel}) are different, that is, $\mathrm{set}(Y_1^m)\neq \mathrm{set}(Y_2^m)$, then clearly the two multisets (as in Definition~\ref{DefMultiModel}) are also different, i.e., $Y_1^m\neq Y_2^m$. Therefore, if we can determine the transmitted word in the non-multiset model by $\mathrm{set}(Y^m)$ obtained from the $N$ channels with different insertion errors, then the same number of channels is enough to determine the transmitted word using $Y^m$ in the multiset model. Thus, if $N_q^{nm}(C;n,t)$ exists, then also $N_q^m(C;n,t)$ exists and \eqref{mLessnm} is true. 
  
 2) Second, we show that from the existence of $N_q^{nm}(n,t)$ the existence of $N_q^{nm}(C;n,t)$ follows.
 By Definition~\ref{DefNonMultiModel}, it is clear that if $N_q^{nm}(n,t)=N_q^{nm}(\mathbb{Z}_q^n;n,t)$ exists, then $N_q^{nm}(C;n,t)$ also exists and $N_q^{nm}(C;n,t)\le N_q^{nm}(n,t)$. Indeed,  if we can determine the transmitted word in $C=\mathbb{Z}_q^n$ with $N$ channels, we can trivially determine the transmitted word in a \emph{subset} $C\subseteq \mathbb{Z}_q^n$ with $N$ channels. 
   
 3) Finally, we show that $N_q^{nm}(n,t)$ exists. Let $N=M_q(n,t)$ (that is, $N$ is equal to the maximum number of different insertion errors given in \eqref{eq:multisetspheresize}) and recall that $\mathrm{set}(M_t(\vec{x}))=I_t(\vec{x})$. By \cite{levenshtein2001efficient}, we know that $N_q^+(n,t)< |I_t(\vec{x})|$ for all words $\vec{x}$. Thus, Theorem~\ref{TheLevInsBound} implies that $I_t(\vec{x})\neq I_t(\vec{x'})$ giving $\mathrm{set}(M_t(\vec{x}))=I_t(\vec{x})\neq I_t(\vec{x'})=\mathrm{set}(M_t(\vec{x'}))$ for all $\vec{x}\neq \vec{x'}$. Consequently, $N_q^{nm}(n,t) < M_q(n,t)$ and, thus, $N_q^{nm}(n,t)$ exists and the existence of $N_q^{nm}(C;n,t)$ and $N_q^m(C;n,t)$ comes from 1) and 2). 

 \smallskip
 
In addition, it is worth noticing that since the parameter $N$ in Definition~\ref{DefMultiModel}  exists for $\mathcal{L}=1$, the corresponding parameter $N$  also trivially exists for every $\mathcal{L}$ where  $\mathcal{L}=\max_{\vec{x}\in C} |T(Y^m)|\ge 1$. The same holds for Definition~\ref{DefNonMultiModel} where $\mathcal{L}=\max_{\vec{x}\in C} |T(Y)|\ge 1$.
 \end{remark}

\subsection{Analysis of \texorpdfstring{$N_q^m(1, t)$}{Nq{m}(1, t)} and \texorpdfstring{$N_q^{nm}(1, t)$}{Nq{nm}(1, t)}}
In order to familiarize ourselves with the concepts, we begin by determining the values $N_q^m(1, t)$ and $N_q^{nm}(1, t)$ (that is, $n = 1$) for all $t$ and $q$. 
\begin{theorem}
\label{thm:nonmultn=1}
    We have $N_q^{nm}(1, t) = (t + 1)(q^t - (q - 1)^t)$.
\end{theorem}
\begin{proof}
      Let $\vec{x}\in\mathbb{Z}_q^1$ and $\vec{x'}\in\mathbb{Z}_q^1$ where $\vec{x}\neq \vec{x'}$. We can assume without loss of generality that $\vec{x} = a$ and $\vec{x'} = b$ for some distinct $a,b\in \mathbb{Z}_q$. Denote by $I$ the set of all insertion vectors of weight $t$ and by $I_b\subseteq I$ (resp. $I_a\subseteq I$) the set of insertion vectors of weight $t$ which do \emph{not} contain the symbol $b$ (resp. the symbol $a$). Suppose that we obtain the word $\vec{x}_b$ from $\vec{x}$ with an insertion vector in $I_b$, then $\vec{x}_b$ does not contain any symbols $b$ and, thus, cannot be obtained from $\vec{x}'$. On the other hand, if we obtain $\vec{w}_b$ from $\vec{x}$ with an insertion vector in $I\setminus I_b$, then $\vec{w}_b$ contains at least one $b$ and can also be obtained from $\vec{x}'$. Similar reasoning holds for $\vec{x}'$ and $I_a$. Hence, $N_q^{nm}(1,t)=|I\setminus I_a|=|I\setminus I_b|$.

    Since $|I_a|=|I_b|=M_{q-1}(1,t)=(q-1)^t\binom{1+t}{t}=(t+1)(q-1)^t$ and $|I|=M_q(1,t)=(t+1)q^t$, we have $N_q^{nm}(1,t)=|I\setminus I_b|=(t+1)(q^t-(q-1)^t)$. 
\end{proof}

\begin{theorem}
\label{thm:multn=1}
For $q \geq 3$, we have $$N_q^m(1,t) = q^t(t + 1) - \sum_{k = 0}^{t}\sum_{m = 0}^{\min(t - k, k)} (q - 2)^{k - m} \binom{t + 1}{k}\binom{k}{m}(t + 1 - k - m),$$
and for the binary case, we have $N_2^m(1,t) = 2^t(t + 1) - \sum_{k = 0}^{\lfloor \frac{t}{2} \rfloor}\binom{t + 1}{k}(t + 1 - 2k)$.
\end{theorem}
\begin{proof}
    Let $\vec{x} = a$ and $\vec{x'} = b$, where $a, b \in \mathbb{Z}_q$, $a \neq b$ and $q \geq 3$. Since $a$ and $b$ are arbitrary, we have $N_q^m(1,t) = |M_t(\vec{x}) \cap M_t(\vec{x'})|$. We know from basic set theory that 
    \begin{equation}
    \label{eq:n=1mults}
    \begin{split}
        M_t(\vec{x}) \cap M_t(\vec{x'}) & = M_t(\vec{x}) \ominus (M_t(\vec{x}) \ominus M_t(\vec{x'}))\\
        \implies |M_t(\vec{x}) \cap M_t(\vec{x'})| & = |M_t(\vec{x})| - |M_t(\vec{x}) \ominus M_t(\vec{x'})|. 
    \end{split}
    \end{equation}
    Let us determine the number of words in $M_t(\vec{x}) \ominus M_t(\vec{x'})$. First, observe that the multiplicity of any word in $M_t(\vec{x})$ (resp. $M_t(\vec{x'})$) depends only on the number of $a$'s (resp. $b$'s) in the word. For instance, the multiplicity of any word with $\ell$ symbols $a$ and $m$ symbols $b$ in $M_t(\vec{x})$ and $M_t(\vec{x'})$ is $\ell$ and $m$, respectively. 

    Let $\vec{y} \in M_t(\vec{x})$ be a word with $t + 1 - k$ symbols $a$, where $0 \leq k \leq t$, and $m$ symbols $b$, where $0 \leq m \leq k$. There are $\binom{t + 1}{t + 1 - k} \binom{k}{m}$ ways to choose the positions of the $a$'s and $b$'s. The remaining positions can be filled by any symbol of $\mathbb{Z}_q \setminus \{a, b\}$. Hence, the number of different words of this type is $\binom{t + 1}{t + 1 - k} \binom{k}{m} (q - 2)^{k - m}$. Since the multiplicity of each such word in $M_t(\vec{x})$ and $M_t(\vec{x'})$ is $t + 1 - k$ and $m$, respectively, the total number of words in $M_t(\vec{x}) \ominus M_t(\vec{x'})$ taking into account the multiplicities is given by $$(q - 2)^{k - m} \binom{t + 1}{t + 1 - k}\binom{k}{m}(t + 1 - k - m).$$ Observe that if $m \geq t + 1 - k$, that is, if the multiplicity of the word in $M_t(\vec{x'})$ is greater than or equal to its multiplicity in $M_t(\vec{x})$, then this word will not exist in $M_t(\vec{x}) \ominus M_t(\vec{x'})$. Hence, we must have $m \leq t - k$.

    To compute the size of $M_t(\vec{x}) \ominus M_t(\vec{x'})$, we sum over all words with $t + 1 - k$ symbols $a$ and $m$ symbols $b$, where $0 \leq m \leq k \leq t$ and $m \leq t - k$. Taking the sum over these limits and substituting into Equation~\eqref{eq:n=1mults}, we obtain 
    \begin{equation*}
    \begin{split}
        |M_t(\vec{x}) \cap M_t(\vec{x'})| & = |M_t(\vec{x})| - \sum_{k = 0}^{t}\sum_{m = 0}^{\min(t - k, k)} (q - 2)^{k - m} \binom{t + 1}{t + 1 - k}\binom{k}{m}(t + 1 - k - m). 
    \end{split}
    \end{equation*}
    Finally, using Equation~\eqref{eq:multisetspheresize}, we get the desired result when $q \geq 3$.

    The proof of the result for the binary case goes analogously.
\end{proof}

\section{Analysis of \texorpdfstring{$N_q^m(n, 1)$}{Nq{m}(n, 1)} and \texorpdfstring{$N_q^{nm}(n, 1)$}{Nq{nm}(n, 1)}}\label{BoundAndExtremalPairs}
In this section, we explore in detail the behavior of word pairs under a single insertion in both the multiset and non-multiset models. To begin with, we determine the exact values of $N_q^m(n, 1)$ and $N_q^{nm}(n, 1)$ and come to the somewhat surprising conclusion that the values are equal for all $n$ and $q$. Next, we show that equality holds not only for the worst-case values, but also for the quantities $N^m_{\vec{x}, \vec{x'}}$ and $N^{nm}_{\vec{x}, \vec{x'}}$ themselves. 
This shows that the two models are interconnected for $t=1$ (see Remark~\ref{samemodel}), but it should be noticed this is not generally the case for $t\ge 2$ (see Example~\ref{NxxExample}). Finally, we give the word pairs for which the number of channels required to distinguish them is the maximum.

\subsection{Exact Values of \texorpdfstring{$N_q^m(n, 1)$}{Nq{m}(n, 1)} and \texorpdfstring{$N_q^{nm}(n, 1)$}{Nq{nm}(n, 1)}}
In this subsection, we analyse in detail the values $N_q^m(n, 1)$ and $N_q^{nm}(n, 1)$ (that is, when $t = 1$) for all $n$ and $q$. Let $\vec{x} = x_1 \dots x_n = \vec{x}_{[1, n]}$ and $\vec{x'} = x_1' \dots x_n' = \vec{x'}_{[1, n]}$ be distinct words in $\mathbb{Z}_q^n$. Let $S \subseteq \mathbb{Z}_q^n$ and $a \in \mathbb{Z}_q$. We define $aS$ as the set of words $\{a\vec{x} \mid \vec{x} \in S\}$.

\begin{table}[h]
\caption{Comparison of the values $N_2^m(n, t)$ (left) and $N_2^{nm}(n, t)$ (right).}
    \label{tab:multvsnonmult}
    \centering
    \begin{tabular}{|c|c|c|c|c|c|}
         \hline
         \backslashbox{$t$}{$n$} & 1 & 2 & 3 & 4 & 5\\
         \hline
        1 & 2 & 2 & 3 & 3 & 4\\
        2 & 6 & 11 & 18 & 25 & 36\\
        3 & 20 & 42 & 86 & 142 & 227\\
        4& 50 & 139 & 324 & 640 & 1118\\
        5& 132 & 414 & 1112 & 2500 & 4850\\
        \hline
    \end{tabular}
    \quad
    \begin{tabular}{|c|c|c|c|c|c|}
        \hline 
        \backslashbox{$t$}{$n$} & 1 & 2 & 3 & 4 & 5\\
         \hline
        1 & 2 & 2 & 3 & 3 & 4\\
        2 & 9 & 14 & 25 & 32 & 48\\
        3 & 28 & 60 & 125 & 196 & 328\\
        4& 75 & 205 & 490 & 910 & 1686\\
        5& 186 & 616 & 1666 & 3570 & 7272\\
        \hline
    \end{tabular}
\end{table}

The main result of this section is that $N^m_q(n, 1)=N^{nm}_q(n, 1)$. However, this is not generally true for $t > 1$
as indicated by Table~\ref{tab:multvsnonmult}, which has been obtained by exhaustive computer searches. Let us first give an upper bound on $N^m_q(n, 1)$ and $N^{nm}_q(n, 1)$.

\begin{theorem}
\label{thm:nonmultt=1upperbd}
    We have $N_q^{m}(n, 1)\le N_q^{nm}(n, 1) \leq \lceil \frac{n + 2}{2} \rceil$.
\end{theorem}
\begin{proof} 
Since $N_q^{m}(n, 1)\le N_q^{nm}(n, 1)$ by \eqref{mLessnm},  it suffices to prove the latter inequality of the claim.
Consider distinct words $\vec{x}, \vec{x'} \in \mathbb{Z}_q^n$. Let $i = h + 1$ be the first coordinate where $\vec{x}$ and $\vec{x'}$ differ, that is, $x_i \neq x_i'$. If they differ already in the first coordinate, then $h = 0$. Similarly, let $n - k$ be the last coordinate where $\vec{x}$ and $\vec{x'}$ differ. Again, if they differ in the $n$th coordinate, then $k = 0$. Notice that the first and the last coordinate where $\vec{x}$ and $\vec{x'}$ differ can also be the same if they differ in exactly one coordinate. Clearly, we have $h + k \leq n - 1$ and $h, k \geq 0$. Let $x_{h + 1} = a$, $x_{h + 1}' = b$, $x_{n - k} = c$ and $x_{n - k}' = d$, where $a \neq b$ and $c \neq d$. When $h + k = n - 1$, $a$ and $c$ (and $b$ and $d$) coincide. For illustrative purposes, we can represent $\vec{x}$ and $\vec{x'}$ as follows:  

{\centering
    \begin{tabular}{c c c c c c c c c c c}
        $\vec{x}$ & $ =$ & $x_1$ & $\dots$ & $x_h$ & $a$ & $\dots$ & $c$ & $x_{n - k + 1}$ & \dots & $x_n$ \\
        $\vec{x'}$ & $ =$ & $x_1$ & $\dots$ & $x_h$ & $b$ & $\dots$ & $d$ & $x_{n - k + 1}$ & $\dots$ & $x_n$\\
    \end{tabular}
    \par
    }
    \noindent since $x_i = x_i'$ for $i \leq h$ and $i \geq n - k + 1$ (when such indices exist, that is $h > 0$ or $k > 0$). Furthermore, we will represent the number of insertions that have occurred with a subscript. That is, $\vec{x} = \vec{x}_{(0)}$ and $\vec{x}_{(1)}$ represent the word obtained by making no, and one insertion(s), respectively, in $\vec{x}$. Notice that in this proof, with $a$, we refer to the specific letter in $\vec{x}_{(0)}$ where $\vec{x}$ and $\vec{x'}$ differ for the first time counting from the left, although the same letter $a$ can appear in other places also (we use similarly the reference to the specific letters $b$, $c$ and $d$).
    Without loss of generality, let us assume that $a \neq x_h$ (if $h > 0$), which is possible since at least one of $a$ or $b$ is not equal to $x_h$.

    Let $A = \mathrm{set}(M_1(\vec{x}) \cap M_1(\vec{x'}))$. 
    Our first goal is to determine an upper bound on the value of $S_{t, \vec{x}}(\vec{x'})$.
    First, observe that if an insertion $f$ occurs before $b$ in $\vec{x'}$, then the resulting word is in $A$ if only if the insertion occurs right before $b$ (see the illustration below):

    {\centering
    \begin{tabular}{c c c c c c c c c c c c}
        $\vec{x}_{(0)}$ & $ =$ & $x_1$ & $\dots$ & $x_h$ & $a$ & $\dots$ & $c$ & $x_{n - k + 1}$ & $\dots$ & $x_n$ \\
        $\vec{x'}_{(1)}$ & $ =$ & $x_1$ & $\dots$ & $x_h$ & $f$ & $b$ & $\dots$ & $d$ & $x_{n - k + 1}$ & $\dots$ & $x_n$\\
    \end{tabular}
    \par
    }
    
    \noindent Indeed, if the insertion occurs elsewhere before $b$ (in the case $h > 0$), then the word cannot be obtained from $\vec{x}$ by a single insertion, because $a \neq x_h$ and $a \neq b$. Moreover, if $f \neq a$, the word $\vec{x'}_{(1)}$ cannot be obtained from $\vec{x}$ by a single insertion. Hence, $f = a$, and the number of words of this form in $A$ is at most 1.

    If no insertion is made to $\vec{x'}$ before $b$, then since $a \neq b$, there has to be an insertion in $\vec{x}$ before $a$, for instance the symbol $b$ just before $a$ (see the illustration below):

    {\centering
    \begin{tabular}{c c c c c c c c c c c c}
        $\vec{x}_{(1)}$ & $ =$ & $x_1$ & $\dots$ & $x_h$ & $b$ & $a$ & $\dots$ & $c$ & $x_{n - k + 1}$ & $\dots$ & $x_n$ \\
        $\vec{x'}_{(0)}$ & $ =$ & $x_1$ & $\dots$ & $x_h$ & $b$ & $\dots$ & $d$ & $x_{n - k + 1}$ & $\dots$ & $x_n$\\
    \end{tabular}
    \par
    }
    \noindent If $b$ and $d$ are distinct (that is $h + k \neq n - 1$), we cannot insert any letters between them, since $c \neq d$. Hence, the insertion can only occur after $d$, and the resultant word could potentially be obtained from $\vec{x}$ as well (see the illustration below):

    {\centering
    \begin{tabular}{c c c c c c c c c c c c}
        $\vec{x}_{(1)}$ & $ =$ & $x_1$ & $\dots$ & $x_h$ & $b$ & $a$ & $\dots$ & $c$ & $x_{n - k + 1}$ & $\dots$ & $x_n$ \\
        $\vec{x'}_{(1)}$ & $ =$ & $x_1$ & $\dots$ & $x_h$ & $b$ & $\dots$ & $d$ & $x_{n - k + 1}$ & $\dots$ & $\dots$ & $x_n$\\
    \end{tabular}
    \par
    }
    
    If $k = 0$ or $c \neq x_{n - k + 1}$, then for $\vec{x'}_{(1)}$ to possibly be in $A$, the insertion after $d$ must be the letter $c$ immediately after $d$. If $k > 0$ and $c = x_{n - k + 1}$, then we can possibly insert in any of the available positions after $d$ to get a word in $A$ (but the inserted symbol is always uniquely determined by the corresponding symbol in $\vec{x}$). There are $k + 1$ such positions after $d$ in $\vec{x'}$, and hence the total number of words of this form in $A$ is at most $k + 1$. Therefore, we have 
    \begin{equation}
        \label{t=1:bound}
    S_{t, \vec{x}}(\vec{x'}) \leq \begin{cases}
        k + 2 & \text{if } c = x_{n - k + 1},\\
        2 & \text{if } c \neq x_{n - k + 1}.
    \end{cases}
    \end{equation}
    when $k > 0$. Observe that the upper bound also holds for $k = 0$.

    \medskip
    Next, we determine an upper bound on $S_{t, \vec{x'}}(\vec{x})$.
    In $\vec{x}$, if an insertion occurs before $a$, then this word can also be obtained from $\vec{x'}$ by inserting a suitable symbol at some position only after $b$ since $a \neq b$ (see the illustration below) :  

    {\centering
    \begin{tabular}{c c c c c c c c c c c c}
        $\vec{x'}_{(0)}$ & $ =$ & $x_1$ & $\dots$ & $x_h$ & $b$ & $\dots$ & $d$ & $x_{n - k + 1}$ & $\dots$ & $x_n$ \\
        $\vec{x}_{(1)}$ & $ =$ & $x_1$ & $\dots$ & $\dots$ & $x_h$ & $a$ & $\dots$ & $c$ & $x_{n - k + 1}$ & $\dots$ & $x_n$\\
    \end{tabular}
    \par
    }
    \noindent Hence, in each of the $h+1$ positions before $a$, the symbol inserted to $\vec{x}$ is uniquely determined by the corresponding symbol of $\vec{x'}$. Thus, we have at most $h + 1$ different insertion vectors of this kind giving words in $A$. This holds also for $h=0$.

    If no insertion is made to $\vec{x}$ before $a$, there has to be an insertion in $\vec{x'}$ before $b$ since $a \neq b$. In particular, it must be the symbol $a$ right before $b$, if $h=0$ and also if $h>0$ since     
    $a \neq x_h$. Thus, we know the word $\vec{x'}$ after the insertions exactly (see the illustration below):
    
    {\centering
    \begin{tabular}{c c c c c c c c c c c c}
        $\vec{x'}_{(1)}$ & $ =$ & $x_1$ & $\dots$ & $x_h$ & $a$ & $b$ & $\dots$ & $d$ & $x_{n - k + 1}$ & \dots & $x_n$ \\
        $\vec{x}_{(0)}$ & $ =$ & $x_1$ & $\dots$ & $x_h$ & $a$ & $\dots$ & $c$ & $x_{n - k + 1}$ & $\dots$ & $x_n$\\
    \end{tabular}
    \par
    }
    If $a$ and $c$ are distinct, we cannot insert any letters between them, since $c \neq d$. Hence, the insertion can only occur after $c$, and the resultant word could potentially be obtained from $\vec{x}$ as well (see the illustration below):

    {\centering
    \begin{tabular}{c c c c c c c c c c c c}
        $\vec{x'}_{(1)}$ & $ =$ & $x_1$ & $\dots$ & $x_h$ & $a$ & $b$ & $\dots$ & $d$ & $x_{n - k + 1}$ & \dots & $x_n$ \\
        $\vec{x}_{(1)}$ & $ =$ & $x_1$ & $\dots$ & $x_h$ & $a$ & $\dots$ & $c$ & $x_{n - k + 1}$ & $\dots$ & $\dots$ & $x_n$\\
    \end{tabular}
    \par
    }
    If $k = 0$ or $d \neq x_{n - k + 1}$, then for $\vec{x}_{(1)}$ to possibly be in $A$, the insertion after $c$ must be the letter $d$ immediately after $c$. If $k > 0$ and $d = x_{n - k + 1}$, then we can possibly insert in any of the available positions after $c$ to get a word in $A$ (but the inserted symbol is always uniquely determined by $\vec{x'}$). There are $k + 1$ such positions after $c$ in $\vec{x}$, and hence the total number of words of this form in $A$ is at most $k + 1$. Thus, we have 
    $$S_{t, \vec{x'}}(\vec{x}) \leq \begin{cases}
        h + k + 2 & \text{if } d = x_{n - k + 1},\\
        h + 2 & \text{if } d \neq x_{n - k + 1}.
    \end{cases}$$
    This bound also holds if $h = 0$ or $k = 0$. 

    Since $c \neq d$, we have by \eqref{t=1:bound} and the inequality above that
\begin{equation}
\label{eq:nonmultubfinaleq}
        \begin{split}
        \min \left(S_{t, \vec{x'}}(\vec{x}), S_{t, \vec{x}}(\vec{x'})\right) & \leq \begin{cases}
\min (h + 2, k + 2) & \text{if } c = x_{n - k + 1}, d \neq x_{n - k + 1},\\
\min (h + k + 2, 2) & \text{if } c \neq x_{n - k + 1}, d = x_{n - k + 1},\\
      \min (h + 2, 2) & \text{if } c \neq x_{n - k + 1}, d \neq x_{n - k + 1}
    \end{cases}\\
    & \leq \min (h + 2, k + 2).
        \end{split}
    \end{equation}
    since, $h, k \geq 0$.
       Therefore, the maximum over any distinct $\vec{x}$ and $\vec{x'}$ of the minimum (on the left-hand side) above is bounded from above by $\max_{h,k} \min (h + 2, k + 2)$.

    Let $n - h - k = m$. Hence, we have $m \geq 1$. Observe that $h + 2$ is a function increasing in $h$, while $k + 2 = n - h - m + 2$ is a function decreasing in $h$. Hence, the maximum value of $\min(h + 2, n - h - m + 2)$ is attained at the point where the two lines intersect, that is, when $h = k = \frac{n - m}{2}$. Since, $h$ and $k$ are integers, let us assume without loss of generality that either $h = \lfloor\frac{n - m}{2}\rfloor$ and $k = \lceil\frac{n - m}{2} \rceil$, or $h = \lceil\frac{n - m}{2}\rceil$ and $k = \lfloor\frac{n - m}{2} \rfloor$. First, let $h = \lfloor\frac{n - m}{2}\rfloor$ and $k = \lceil\frac{n - m}{2} \rceil$.
    Since $\min (h + 2, k + 2)$ is a decreasing function in $m$, the maximum value of $\min (h + 2, k + 2)$ is attained when $m$ is minimized. Consider the following cases:

    \begin{enumerate}
        \item When $m = 1$:
        \begin{itemize}
            \item If $n$ is odd, then $h = k = \frac{n - 1}{2}$ and $\min(h + 2, k + 2) = \frac{n + 3}{2}$.
            \item If $n$ is even, then $h = \frac{n - 2}{2}$, $k = \frac{n}{2}$ and $\min(h + 2, k + 2) = \frac{n + 2}{2}$.
        \end{itemize}
        \item When $m = 2$:
        \begin{itemize}
            \item If $n$ is odd, then $h = \frac{n - 3}{2}$, $k = \frac{n - 1}{2}$ and $\min(h + 2, k + 2) = \frac{n + 1}{2}$.
            \item If $n$ is even, then $h = k = \frac{n - 2}{2}$ and $\min(h + 2, k + 2) = \frac{n + 2}{2}$.
        \end{itemize}
    \end{enumerate}
    For $m > 2$, $\min(h + 2, k + 2)$ will clearly be less than $\frac{n + 2}{2}$. Therefore, $N_q^{nm}(n, 1) \leq \lceil \frac{n + 2}{2} \rceil$ and the maximum value is attained when $n$ is even and $h = \frac{n - 2}{2}$, $k = \frac{n}{2}$ or $h = k = \frac{n - 2}{2}$; or when $n$ is odd and $h = k = \frac{n - 1}{2}$. Similarly, by considering the case when $h = \lceil\frac{n - m}{2}\rceil$ and $k = \lfloor\frac{n - m}{2} \rfloor$, we obtain that $N_q^{nm}(n, 1) \leq \lceil \frac{n + 2}{2} \rceil$ and the maximum value is attained when $n$ is even and $h = \frac{n}{2}$, $k = \frac{n - 2}{2}$ or $h = k = \frac{n - 2}{2}$; or when $n$ is odd and $h = k = \frac{n - 1}{2}$.
\end{proof}

In what follows, we aim to derive an optimal lower bound for $N_q^{m}(n, 1)$ that attains the upper bound of the previous theorem; this will be established in Theorem~\ref{thm:lowerboundt=1odd}, where we present two words for which the intersection has size $\lceil (n+2)/2 \rceil$. By \eqref{mLessnm}, this immediately yields an optimal lower bound for $N_q^{nm}(n,1)$ as well. To begin, we first develop a better understanding of the words in the multiset $M_1(\vec{x})$. To this end, we illustrate the $1$-insertion multiset spheres $M_1(\vec{x})$ and $M_1(\vec{x'})$ in Table~\ref{tab:doublerecursion}. Furthermore, assuming $\vec{x} = x_1 \dots x_n = \vec{x}_{[1, n]}$ and $\vec{x'} = x_1' \dots x_n' = \vec{x'}_{[1, n]}$ for $n\ge 3$, we first give the following recurrence relation, which we obtain by partitioning the words in $M_1(\vec{x})$ into words with the same first and the last letter: 
\begin{equation}
    \label{eq:doubrec}
    \begin{split}
       M_1(\vec{x}) = x_1M_1(\vec{x}_{[2, n - 1]})x_n \underset{a \in \mathbb{Z}_q}{\bigoplus} x_1\vec{x}_{[2, n]}a \underset{b \in \mathbb{Z}_q}{\bigoplus} b\vec{x}_{[1, n - 1]}x_n.
    \end{split}
\end{equation}
Here the part $ x_1M_1(\vec{x}_{[2, n - 1]})x_n$ corresponds to the words where the insertion occurs anywhere between the first letter $x_1$ and the last letter $x_n$ in $\vec{x}$. In the second part $x_1\vec{x}_{[2, n]}a$ (resp. in the last part $b\vec{x}_{[1, n - 1]}x_n$), the insertion occurs at the end (resp. at the beginning) of $\vec{x}$.
We can write an analogous expression for $M_1(\vec{x'})$, and depending on the relation between $x_1$ and $x_1'$, and $x_n$ and $x_n'$, we have different cases for $M_1(\vec{x}) \cap M_1(\vec{x'})$. For instance, if $x_1=x_1'$ and $x_n=x_n'$, then we apply~\eqref{eq:doubrec} and group the words by their first and last letters to obtain:
\begin{equation}
 \begin{split}
 \label{eq:doublerecx1=x1'xn=xn'}
   M_1(\vec{x}) \cap M_1(\vec{x'}) = &\big[(x_1M_1(\vec{x}_{[2, n - 1]})x_n \oplus x_1\vec{x}_{[2, n]}x_n \oplus x_1\vec{x}_{[1, n - 1]}x_n) \\& \cap (x_1M_1(\vec{x'}_{[2, n - 1]})x_n \oplus x_1\vec{x'}_{[2, n]}x_n \oplus x_1\vec{x'}_{[1, n - 1]}x_n)\big] \\& \underset{\substack{a \in \mathbb{Z}_q \\ a \neq x_n}}{\bigoplus} \big[x_1\vec{x}_{[2, n]}a \cap x_1\vec{x'}_{[2, n]}a] \underset{\substack{b \in \mathbb{Z}_q \\ b \neq x_1}}{\bigoplus} [b\vec{x}_{[1, n - 1]}x_n \cap b\vec{x'}_{[1, n - 1]}x_n\big].
 \end{split}   
\end{equation}

\begin{table}
\caption{A representation of the recursion relation for insertion spheres $M_1(\vec{x})$ and $M_1(\vec{x'})$ in the multiset model. Here we separate the words in $M_1(\vec{x})$ first depending on whether the insertion occurs in the middle of the word $\vec{x}$ or before the first letter or after the last letter of $\vec{x}$. Moreover, we separate the cases based on whether the insertion to the end (resp. beginning) is $x_n$ (resp. $x_1$) itself or not. The same is done for $M_1(\vec{x'})$.}
\parbox{.45\linewidth}{
\centering
\begin{tabular}{|c|c|c|}
\hline
$x_1$ &  & $x_n$\\
\vdots & $M_1(\vec{x}_{[2, n - 1]})$ & \vdots\\
$x_1$ & & $x_n$\\
\hline
\hline
$x_1$ & $\vec{x}_{[2, n]}$ & $0$\\
\hline
\vdots & \vdots & \vdots\\
\hline
$x_1$ & $\vec{x}_{[2, n]}$ & $x_n$\\
\hline
\vdots & \vdots & \vdots\\
\hline
$x_1$ & $\vec{x}_{[2, n]}$ & $q-1$\\
\hline
\hline
$0$ & $\vec{x}_{[1, n - 1]}$ & $x_n$\\
\hline
\vdots & \vdots & \vdots\\
\hline
$x_1$ & $\vec{x}_{[1, n - 1]}$ & $x_n$\\
\hline
\vdots & \vdots & \vdots\\
\hline
$q - 1$ & $\vec{x}_{[1, n - 1]}$ & $x_n$\\
\hline
\end{tabular}
}
\hfill
\parbox{.45\linewidth}{
\centering
\begin{tabular}{|c|c|c|}
\hline
$x_1'$ &  & $x_n'$\\
\vdots & $M_1(\vec{x'}_{[2, n - 1]})$ & \vdots\\
$x_1'$ & & $x_n'$\\
\hline
\hline
$x_1'$ & $\vec{x'}_{[2, n]}$ & $0$\\
\hline
\vdots & \vdots & \vdots\\
\hline
$x_1'$ & $\vec{x'}_{[2, n]}$ & $x_n'$\\
\hline
\vdots & \vdots & \vdots\\
\hline
$x_1'$ & $\vec{x'}_{[2, n]}$ & $q - 1$\\
\hline
\hline
$0$ & $\vec{x'}_{[1, n - 1]}$ & $x_n'$\\
\hline
\vdots & \vdots & \vdots\\
\hline
$x_1'$ & $\vec{x'}_{[1, n - 1]}$ & $x_n'$\\
\hline
\vdots & \vdots & \vdots\\
\hline
$q - 1$ & $\vec{x'}_{[1, n - 1]}$ & $x_n'$\\
\hline
\end{tabular}
}
\label{tab:doublerecursion}
\end{table}

Equation~\eqref{eq:doublerecx1=x1'xn=xn'} helps us to obtain recursive relations for the size of the intersection of $M_1(\vec{x})$ and $M_1(\vec{x'})$ as shown by the following results.

\begin{proposition}
\label{prop:x1=x1'xn=xn'}
Let $n\ge 3$ and $\vec{x}, \vec{x'} \in \mathbb{Z}_q^n$ be such that $\vec{x} \neq \vec{x'}$ with $x_1 = x_1'$ and $x_n = x_n'$. We have
\begin{equation*}
    M_1(\vec{x}) \cap M_1(\vec{x'}) = \big[x_1M_1(\vec{x}_{[2, n - 1]})x_n \oplus \vec{x}x_n \oplus x_1\vec{x}\big] \cap \big[x_1M_1(\vec{x'}_{[2, n - 1]})x_n \oplus \vec{x'}x_n \oplus x_1\vec{x'}\big].
\end{equation*}
\end{proposition}
\begin{proof}
    In Equation~\eqref{eq:doublerecx1=x1'xn=xn'}, observe that if either $x_1\vec{x}_{[2, n]} a = x_1\vec{x'}_{[2, n]} a$ or $b\vec{x}_{[1, n - 1]}x_n = b\vec{x'}_{[1, n - 1]}x_n$ for any $a,b \in \mathbb{Z}_q$, this implies that $\vec{x} = \vec{x'}$, since $x_1\vec{x}_{[2, n]} = \vec{x}_{[1, n - 1]}x_n = \vec{x}$ and $x_1\vec{x'}_{[2, n]} = \vec{x'}_{[1, n - 1]}x_n = \vec{x'}$. Hence, $\{x_1\vec{x}_{[2, n]}a\} \cap \{x_1\vec{x'}_{[2, n]}a\} = \emptyset$ and $\{b\vec{x}_{[1, n - 1]}x_n\} \cap \{b\vec{x'}_{[1, n - 1]}x_n\} = \emptyset$. Therefore, Equation~\eqref{eq:doublerecx1=x1'xn=xn'} can be simplified to get 
\begin{equation}
\label{eq:propx1=x1'xn=xn'proof2}
 \begin{split}
    M_1(\vec{x}) \cap M_1(\vec{x'}) = & (x_1M_1(\vec{x}_{[2, n - 1]})x_n \oplus x_1\vec{x}_{[2, n]}x_n \oplus x_1\vec{x}_{[1, n - 1]}x_n) \\& \cap (x_1M_1(\vec{x'}_{[2, n - 1]})x_n \oplus x_1\vec{x'}_{[2, n]}x_n \oplus x_1\vec{x'}_{[1, n - 1]}x_n).
 \end{split} 
\end{equation}
This completes the proof.
\end{proof}

\begin{corollary}
\label{cor:x1=x1'xn=xn'}
    Let $n\ge 3$  and  $\vec{x}, \vec{x'} \in \mathbb{Z}_q^n$ be such that $\vec{x} \neq \vec{x'}$ with $x_1 = x_1'$ and $x_n = x_n'$. We get $|M_1(\vec{x}) \cap M_1(\vec{x'})| \leq |M_1(\vec{x}_{[2, n - 1]}) \cap M_1(\vec{x'}_{[2, n - 1]})| + 4$.
\end{corollary}
\begin{proof}
We take the cardinalities of sets on both sides of Equation (\ref{eq:propx1=x1'xn=xn'proof2}). In the process, we can discard the initial $x_1$ and the final $x_n$ since they are common to all the words in the right hand side and we are interested in the number of words rather than the words themselves. Clearly, in this case $|M_1(\vec{x}) \cap M_1(\vec{x'})| \leq |M_1(\vec{x}_{[2, n - 1]}) \cap M_1(\vec{x'}_{[2, n - 1]})| + 4$ since each of the words $\vec{x}x_n$, $x_1\vec{x}$, $\vec{x'}x_n$ and $x_1\vec{x'}$ can contribute at most once to the intersection.
\end{proof}

With more careful analysis, it is possible to improve the previous bound as shown in the next theorem (the proof is somewhat technical and it is deferred to Appendix). The bound in the next theorem is actually tight; indeed, for example, the word pairs given in~Theorem~\ref{thm:lowerboundt=1odd} attain the inequality.

\begin{theorem}
\label{thm:doublerecupperbound}
    Let $n\ge 3$ and $\vec{x}, \vec{x'} \in \mathbb{Z}_q^n$ be such that $\vec{x} \neq \vec{x'}$ with $x_1 = x_1'$ and $x_n = x_n'$. We have
    \begin{equation}\label{+1Bound}
           |M_1(\vec{x}) \cap M_1(\vec{x'})| \leq |M_1(\vec{x}_{[2, n - 1]}) \cap M_1(\vec{x'}_{[2, n - 1]})| + 1.
           \end{equation}
    \end{theorem}
\begin{proof}
    See Appendix.
\end{proof}

From Proposition~\ref{prop:x1=x1'xn=xn'} it follows directly for distinct $\vec{x}$ and $\vec{x'}$ with $x_1 = x_1'$ and $x_n = x_n'$ that
\begin{equation}\label{+0Bound}
    |M_1(\vec{x}) \cap M_1(\vec{x'})|\ge |M_1(\vec{x}_{[2, n - 1]}) \cap M_1(\vec{x'}_{[2, n - 1]})|,
\end{equation}
 so there is at most the difference of one between the sizes of the two intersections. In Remark~\ref{RmkCharacterization} in Appendix, we explain a way to characterize word pairs $\vec{x}$ and $\vec{x'}$ for which \eqref{+0Bound} (resp. the upper bound \eqref{+1Bound}) is attained.
Having examined the case where $x_1 = x_1'$ and $x_n = x_n'$, we now consider in the following results the situations where $x_1 \neq x_1'$ or $x_n \neq x_n'$.

\begin{proposition}
\label{prop:x1=x1xn!=xn}
Let 
$\vec{x}, \vec{x'} \in \mathbb{Z}_q^n$ be such that 
$x_n \neq x_n'$. We have $M_1(\vec{x}) \cap M_1(\vec{x'}) \subseteq \{\vec{x'}_{[1, n]}x_n, \vec{x}_{[1, n]}x_n'\}$. In particular, $N_{\vec{x}, \vec{x'}}^m = |M_1(\vec{x}) \cap M_1(\vec{x'})| \leq 2$.
\end{proposition}
\begin{proof}
Let $\vec{x}, \vec{x'} \in \mathbb{Z}_q^n$ be such that 
$x_n \neq x_n'$. Since $x_n \neq x_n'$, the insertions in $\vec{x}$ and $\vec{x'}$ cannot both be before or after the $n$th coordinate. The word obtained by inserting to the right of $x_n$ in $\vec{x}$ can possibly be obtained from $\vec{x'}$ only if the insertion in $\vec{x}$ is the symbol $x_n'$, and in this case the resulting word is of the form $\vec{x}_{[1, n]}x_n'$. Since $x_n \neq x_n'$, $m(\vec{x}_{[1, n]}x_n', M_1(\vec{x})) = 1$. Similarly, the only word obtained from $\vec{x'}$ by inserting to the right of $x_n'$ that is possibly in $M_1(\vec{x})$ is of the form $\vec{x'}_{[1, n]}x_n$ and $m(\vec{x'}_{[1, n]}x_n, M_1(\vec{x'})) = 1$. Therefore, $M_1(\vec{x}) \cap M_1(\vec{x'}) \subseteq \{\vec{x'}_{[1, n]}x_n, \vec{x}_{[1, n]}x_n'\}$, and $N_{\vec{x}, \vec{x'}}^m \leq 2$.
\end{proof}

The case when $x_1 \neq x_1'$ behaves similarly to the case when $x_n \neq x_n'$ in Proposition~\ref{prop:x1=x1xn!=xn}. This is due to the fact that if we reverse a word with the property $x_n \neq x_n'$, we obtain a word with the property $x_1 \neq x_1'$. The following results formally state this simple observation, which turns out to be very useful in the sequel.
\begin{proposition}
\label{lem:reverse}
    The function $f: \mathbb{Z}_q^n \rightarrow \mathbb{Z}_q^n$ which maps a word $\vec{x} = x_1 \dots x_n$ to its reverse $\overleftarrow{\vec{x}} = x_n \dots x_1$ is bijective.
\end{proposition}

\begin{corollary}
\label{cor:reverse}
    If $\vec{a} \in \mathbb{Z}_q^n$, then $m(\vec{a}, M_t(\vec{x})) = m(\overleftarrow{\vec{a}}, M_t(\overleftarrow{\vec{x}}))$.
\end{corollary}
\begin{proof}
    If we can obtain $\vec{a}$ from $\vec{x}$ with insertion vector $w = (\vec{w_0}, \dots, \vec{w_{n}})$, then we can obtain $\overleftarrow{\vec{a}}$ from $\overleftarrow{\vec{x}}$ with insertion vector $w' = (\overleftarrow{\vec{w_{n}}}, \dots, \overleftarrow{\vec{w_0}})$. Now the result follows by Proposition~\ref{lem:reverse}.
\end{proof}

Now we can obtain an upper bound on the size of $M_1(\vec{x}) \cap M_1(\vec{x'})$ when  $x_1 \neq x_1'$. 
\begin{corollary}
\label{cor:x1!=x1xn=xn}
    If $\vec{x}, \vec{x'} \in \mathbb{Z}_q^n$ are such that $x_1 \neq x_1'$, then $N_{\vec{x}, \vec{x'}}^m \leq 2$.
\end{corollary}
\begin{proof}
    The result is obtained directly from Proposition~\ref{prop:x1=x1xn!=xn} and Corollary~\ref{cor:reverse}.
\end{proof}

Next, we give a result concerning words of short lengths, which we need in the proof of Theorem~\ref{thm:lowerboundt=1odd} (and also later in Theorem~\ref{thm:extremalmulteven}).

\begin{observation}
\label{obs:t=1upperbound}
    If the lengths of both $\vec{x}$ and $\vec{x'}$ are $1$ or $2$, and $\vec{x}\neq \vec{x'}$, then $|M_1(\vec{x}) \cap M_1(\vec{x'})| \leq 2$ (by Proposition~\ref{prop:x1=x1xn!=xn} and Corollary~\ref{cor:x1!=x1xn=xn}). Moreover, any word pair $\vec{x}$ and $\vec{x'}$ of length $2$ for which $|M_1(\vec{x}) \cap M_1(\vec{x'})| = 2$ is of the form $\{aa,ab\}$, $\{aa,ba\}$, $\{ab,ba\}$, $\{ab,ac\}$ or $\{ab,cb\}$, for some distinct $a,b,c\in \mathbb{Z}_q$ if $q\ge 3$ and is of the form $\{aa,ab\}$, $\{aa,ba\}$ or $\{ab,ba\}$, for distinct $a, b \in \mathbb{Z}_q$ if $q=2$. 
\end{observation}
\begin{proof}
    Let $a$ and $b$  be arbitrary distinct letters in $\mathbb{Z}_q$. If $n = 1$, then $M_1(a) \cap M_1(b) = \{ab, ba\}$. Next, let $n = 2$. Without loss of generality, we can assume that $\vec{x}=aa$ or $\vec{x}=ab.$
    
    {\bf Case 1:} Suppose first that $\vec{x}=aa$. If $M_1(\vec{x}) \cap M_1(\vec{x'})\neq \emptyset$, then clearly $\vec{x'}$ must contain at least one symbol $a$ (as a single insertion can add at most one $a$). Moreover, $\vec{x'}$ contains at most one $a$ since $\vec{x}\neq \vec{x'}$. Therefore, without loss of generality, $\vec{x'}=ab$ or $\vec{x'}=ba$. Suppose first that $\vec{x'}=ab$. Then $M_1(\vec{x}) \cap M_1(\vec{x'}) = \{aab, aba\}$.  Suppose then that $\vec{x'}=ba$. Since $\overleftarrow{\vec{x}}=aa$ and $\overleftarrow{\vec{x'}}=ab$, it follows from  Corollary~\ref{cor:reverse} that $|M_1(\vec{x}) \cap M_1(\vec{x'})|= |M_1(\overleftarrow{\vec{x}}) \cap M_1(\overleftarrow{\vec{x'}})| =2$.

    {\bf Case 2:} Assume next that $\vec{x}=ab$. As $t=1$, the word $\vec{x'}$ must have at least one $a$ or one $b$ (and due to Case 1, we can assume that they do not appear more than once). If it has both of them, then $\vec{x'}=ba$ as $\vec{x}\neq \vec{x'}$. Now $M_1(\vec{x}) \cap M_1(\vec{x'}) =\{aba, bab\}.$ Suppose next that $\vec{x'}$ has either $a$ or $b$ but not both; hence, we must have $q\ge 3$ from now on. Assume first that $\vec{x'}$ has a symbol $a$. Then $\vec{x'}$ is $ac$ or $ca$. If $\vec{x'}=ac$, then $M_1(\vec{x}) \cap M_1(\vec{x'}) =\{abc, acb\}.$ On the other hand, if $\vec{x'}=ca$, then $M_1(\vec{x}) \cap M_1(\vec{x'}) =\{cab\}.$ Assume next that $\vec{x'}$ has a symbol $b$. Consequently, $\vec{x'}$ is $bc$ or $cb$. If $\vec{x'}=bc$, then $M_1(\vec{x}) \cap M_1(\vec{x'}) =\{abc\}$. On the other hand, if $\vec{x'}=cb$, then $M_1(\vec{x}) \cap M_1(\vec{x'}) =\{cab,acb\}$.
\end{proof}

For brevity, we write a repeated letter using exponent notation: for a symbol $x$ and an integer $m\ge 0$, $x^{m}$ denotes the word consisting of $m$ consecutive copies of $x$ (with $x^{0}$ being the empty word). Thus, $a^{k}b^{\ell}$ denotes the word consisting of $k$ consecutive $a$’s followed by $\ell$ consecutive $b$’s, and $b^{k}a^{0}b^{\ell}$ simply indicates that no $a$ occurs between the two blocks of $b$’s.
Finally, in the following theorem, we are ready to prove that $N_q^{m}(n, 1) \geq \lceil (n+2)/2 \rceil$, that is, there exist two words $\vec{x}, \vec{x'} \in \mathbb{Z}_q^n$ such that $|M_1(\vec{x}) \cap M_1(\vec{x'})| = \lceil (n+2)/2 \rceil$.
\begin{theorem}
\label{thm:lowerboundt=1odd}
    We have $\biggl|M_1\left(a^{\left\lfloor\frac{n + 2}{2} \right\rfloor}b^{\left\lceil\frac{n - 2}{2}\right\rceil}\right) \cap M_1\left(a^{\left\lfloor\frac{n}{2} \right\rfloor}b^{\left\lceil\frac{n}{2}\right\rceil}\right)\biggr| =  \left\lceil\frac{n + 2}{2}\right\rceil$, where $a, b \in \mathbb{Z}_q$ and $a \neq b$.
\end{theorem}
\begin{proof}
    The proof is based on induction on $n$. When $n = 1$, we know from Observation~\ref{obs:t=1upperbound} that $|M_1(a) \cap M_1(b)| = 2 = \frac{1 + 1}{2} + 1$. When $n = 2$, we know from Observation~\ref{obs:t=1upperbound} that $|M_1(aa) \cap M_1(ab)| = 2 = \frac{2}{2} + 1$. This proves the base case.

    Let us assume that the result holds for any integer less than $n$. We wish to prove that the result is true for $n\ge 3$ as well. Furthermore, for notational simplicity, we choose $a = 0$ and $b = 1$ as the other cases go exactly the same way due to the permutations of symbols. Applying Proposition~\ref{prop:x1=x1'xn=xn'} to $\vec{x} = 0^{\left\lfloor\frac{n + 2}{2} \right\rfloor}1^{\left\lceil\frac{n - 2}{2}\right\rceil}$ and $\vec{x'} = 0^{\left\lfloor\frac{n}{2} \right\rfloor}1^{\left\lceil\frac{n}{2}\right\rceil}$, we get
    \begin{align*}
            M_1\left(0^{\left\lfloor\frac{n + 2}{2} \right\rfloor}1^{\left\lceil\frac{n - 2}{2}\right\rceil}\right) \cap M_1\left(0^{\left\lfloor\frac{n}{2} \right\rfloor}1^{\left\lceil\frac{n}{2}\right\rceil}\right) = &\left(0 M_1\left(0^{\left\lfloor\frac{n}{2} \right\rfloor}1^{\left\lceil\frac{n - 4}{2}\right\rceil}\right) 1 \oplus 0 0^{\left\lfloor\frac{n}{2} \right\rfloor}1^{\left\lceil\frac{n - 2}{2}\right\rceil} 1\oplus 0 0^{\left\lfloor\frac{n + 2}{2} \right\rfloor}1^{\left\lceil\frac{n - 4}{2}\right\rceil} 1\right) \\& \cap \left(0 M_1\left(0^{\left\lfloor\frac{n - 2}{2} \right\rfloor}1^{\left\lceil\frac{n - 2}{2}\right\rceil}\right)1 \oplus 0 0^{\left\lfloor\frac{n - 2}{2} \right\rfloor}1^{\left\lceil\frac{n}{2}\right\rceil} 1 \oplus 0 0^{\left\lfloor\frac{n}{2} \right\rfloor}1^{\left\lceil\frac{n - 2}{2}\right\rceil} 1\right).
    \end{align*}  
    Observe that $0 0^{\left\lfloor\frac{n}{2} \right\rfloor}1^{\left\lceil\frac{n - 2}{2}\right\rceil} 1$ exists in the intersection and hence, by taking the cardinality on both sides, we get,
    \begin{align*}
    \left|M_1\left(0^{\left\lfloor\frac{n + 2}{2} \right\rfloor}1^{\left\lceil\frac{n - 2}{2}\right\rceil}\right) \cap M_1\left(0^{\left\lfloor\frac{n}{2} \right\rfloor}1^{\left\lceil\frac{n}{2}\right\rceil}\right)\right| = & \biggl|\left(M_1\left(0^{\left\lfloor\frac{n}{2} \right\rfloor}1^{\left\lceil\frac{n - 4}{2}\right\rceil}\right) \oplus 0^{\left\lfloor\frac{n + 2}{2} \right\rfloor}1^{\left\lceil\frac{n - 4}{2}\right\rceil}\right) \\ & \cap \left(M_1\left(0^{\left\lfloor\frac{n - 2}{2} \right\rfloor}1^{\left\lceil\frac{n - 2}{2}\right\rceil}\right) \oplus 0^{\left\lfloor\frac{n - 2}{2} \right\rfloor}1^{\left\lceil\frac{n}{2}\right\rceil}\right)\biggr| + 1.
    \end{align*}
    Furthermore, we see that $0^{\left\lfloor\frac{n - 2}{2} \right\rfloor}1^{\left\lceil\frac{n}{2}\right\rceil} \notin M_1\left(0^{\left\lfloor\frac{n}{2} \right\rfloor}1^{\left\lceil\frac{n - 4}{2}\right\rceil}\right)$ and $0^{\left\lfloor\frac{n + 2}{2} \right\rfloor}1^{\left\lceil\frac{n - 4}{2}\right\rceil} \notin M_1\left(0^{\left\lfloor\frac{n - 2}{2} \right\rfloor}1^{\left\lceil\frac{n - 2}{2}\right\rceil}\right)$. Hence, we can simplify the above equation to get 
    \begin{equation*}
    \begin{split}
        \left|M_1\left(0^{\left\lfloor\frac{n + 2}{2} \right\rfloor}1^{\left\lceil\frac{n - 2}{2}\right\rceil}\right) \cap M_1\left(0^{\left\lfloor\frac{n}{2} \right\rfloor}1^{\left\lceil\frac{n}{2}\right\rceil}\right)\right| & = \left|M_1\left(0^{\left\lfloor\frac{n}{2} \right\rfloor}1^{\left\lceil\frac{n - 4}{2}\right\rceil}\right) \cap M_1\left(0^{\left\lfloor\frac{n - 2}{2} \right\rfloor}1^{\left\lceil\frac{n - 2}{2}\right\rceil}\right)\right| + 1 \\
        & = \left|M_1\left(0^{\left\lfloor\frac{(n - 2) + 2}{2} \right\rfloor}1^{\left\lceil\frac{(n - 2) - 2}{2}\right\rceil}\right) \cap M_1\left(0^{\left\lfloor\frac{n - 2}{2} \right\rfloor}1^{\left\lceil\frac{n - 2}{2}\right\rceil}\right)\right| + 1.
    \end{split}
    \end{equation*}
    The result now follows from the induction hypothesis.
\end{proof}

Now we are ready to state the main result in this section, which reveals the somewhat surprising fact that $N^m_q(n, 1)=N^{nm}_q(n, 1)$.

\begin{corollary}
    \label{cor:exactboundt=1}
    We have $N_q^m(n, 1) = N_q^{nm}(n, 1) = \lceil \frac{n + 2}{2} \rceil$.
\end{corollary}
\begin{proof}
    By Theorem~\ref{thm:nonmultt=1upperbd}, $N^{nm}_q(n, 1) \leq \lceil \frac{n + 2}{2} \rceil$. By Theorem~\ref{thm:lowerboundt=1odd},  we get $N_q^m(n, 1) \geq \lceil \frac{n + 2}{2} \rceil$. Consequently, $\lceil \frac{n + 2}{2} \rceil \leq N_q^m(n, 1) \leq N_q^{nm}(n, 1) \leq \lceil \frac{n + 2}{2} \rceil$, and this gives us the required result.
\end{proof}

\subsection{Extremal Word Pairs}

We call a pair of words $\vec{x}$ and $\vec{x}'$ an \emph{extremal word pair} (sometimes also called a \emph{worst word pair}) for the multiset model if they require the maximum number of channels to distinguish them, that is, $N^m_{\vec{x},\vec{x'}} =|M_t(\vec{x}) \cap M_t(\vec{x'})| = N_q^m(n, t)$. Note that there may be multiple such extremal word pairs, as indicated by Table~\ref{tab:extremalwords}. Similarly, a word pair is an extremal word pair for the non-multiset model if $N^{nm}_{\vec{x},\vec{x'}} = N_q^{nm}(n, t)$.
\begin{table}[htbp]
    \centering
    \caption{All the extremal word pairs in the multiset model (up to interchanging 0's and 1's) when $q = 2$.}
    \label{tab:extremalwords}
    \begin{tabular}{|c|c|c|c|c|c|c|c|c|}
    \hline
        \backslashbox{$t$}{$n$} & 1 & 2 & 3 & 4 & 5 & 6 & 7 \\
         \hline
       1 & 0, 1 & 00, 01 & 001, 011 & 0001, 0011 & 00011, 00111 & 000011, 000111 & 0000111, 0001111\\
        &  & 01, 10 &  & 0010, 0100 & & 000100, 001000 & \\
        &  &  &  & 0010, 0110 & & 000110, 001110 & \\
        \hline
       2 & 0, 1 & 00, 01 & 001, 011 & 0001, 0011 & 00011, 00111 & 000100, 001000 & 0000111, 0001111\\
        &  &  &  & 0010, 0100 & & &\\
        \hline
       3 & 0, 1 & 00, 01 & 001, 011 & 0010, 0100 & 00011, 00111 & 000100, 001000 & 0000111, 0001111\\
       \hline
       4 & 0, 1 & 00, 01 & 001, 011 & 0010, 0100 & 00011, 00111 & 000100, 001000 & 0000100, 0001000\\
       \hline
       5 & 0, 1 & 00, 01 & 001, 011 & 0010, 0100 & 00010, 00100 & 000100, 001000 & 0000100, 0001000\\
        \hline
    \end{tabular}
\end{table}

Next, we completely classify the extremal word pairs for $t=1$, for all $n$ and $q$, in both the multiset and non-multiset models. The proof is based on the surprising fact stating that the extremal word pairs are the same for both models when $t=1$! However, note that for general $t$ it is not necessarily true that extremal word pairs in the multiset model will be extremal word pairs in the non-multiset model and vice versa. For instance, it can be verified by computer that when $n = t = 5$ and $q = 2$, $\vec{x} = 00010$ and $\vec{x'} = 00100$ is an extremal word pair in the multiset model but not in the non-multiset model, while $\vec{x} = 00011$ and $\vec{x'} = 00111$ is an extremal word pair in the non-multiset model but not in the multiset model.

In the following theorem, we first show that the number of channels required to distinguish a given word pair is equal in both models.
\begin{theorem}\label{sameproblem}
    Let 
    $\vec{x}, \vec{x'} \in \mathbb{Z}_q^n$. If $t = 1$, then $N^m_{\vec{x}, \vec{x'}} = N^{nm}_{\vec{x}, \vec{x'}}$ for every  $\vec{x}\neq \vec{x'}$.
\end{theorem}
\begin{proof}
    Let $n \geq 1$, $t = 1$, and $\vec{x}, \vec{x'} \in \mathbb{Z}_q^n$ where $\vec{x}\neq \vec{x'}$. Denote $A = M_1(\vec{x}) \cap M_1(\vec{x'})$. Recall that $\mathrm{set}(M_1(\vec{x}))=I_1(\vec{x})$ and $N_q^+(n,1)=\max_{\vec{x},\vec{x'}\in\mathbb{Z}_q^n,\vec{x}\neq \vec{x'}}|I_1(\vec{x})\cap I_1(\vec{x'})|$. Since $N_q^+(n,1)=2$,  we have by~Theorem~\ref{OrigLevIns} that $|\mathrm{set}(A)| \leq 2$. Depending on the size of $|\mathrm{set}(A)|$, we divide the proof into the following two cases.

    Case 1: Assume that $|\mathrm{set}(A)| \leq 1$. If $|\mathrm{set}(A)| = 0$, then the multiset spheres have no common words yielding $N^m_{\vec{x}, \vec{x'}} = 0=N^{nm}_{\vec{x}, \vec{x'}}$. If $|\mathrm{set}(A)| = 1$, that is, $\mathrm{set}(A) = \{\vec{a}\}$ for some $\vec{a}\in \mathbb{Z}_q^{n+1}$, then by Equations~\eqref{Nmxx} and \eqref{Nnmxx}, $N^m_{\vec{x}, \vec{x'}}$ and $N^{nm}_{\vec{x}, \vec{x'}}$ are both equal to $\min (m(\vec{a}, M_1(\vec{x})), m(\vec{a}, M_1(\vec{x'})))$, and hence the result holds.

    Case 2: Assume that $|\mathrm{set}(A)| = 2$. Let $\mathrm{set}(A) = \{\alpha, \beta\}$ for some $\alpha,\beta\in \mathbb{Z}_q^{n+1}$ where $\alpha\neq\beta$. Depending on the multiplicities of $\alpha$ and $\beta$ in $M_1(\vec{x})$ and $M_1(\vec{x'})$, we consider the following subcases:

    Case 2A: Let $m(\vec{a}, M_1(\vec{x})) \leq m(\vec{a}, M_1(\vec{x'}))$ for all $\vec{a} \in \{\alpha, \beta\}$, or $m(\vec{a}, M_1(\vec{x})) \geq m(\vec{a}, M_1(\vec{x'}))$ for all $\vec{a} \in \{\alpha, \beta\}$. We can assume, without loss of generality, that $m(\vec{a}, M_1(\vec{x})) \leq m(\vec{a}, M_1(\vec{x'}))$ for all $\vec{a} \in \{\alpha, \beta\}$. In this case, by Equations~\eqref{Nmxx} and \eqref{Nnmxx}, we have
    \begin{align*}
        N^m_{\vec{x}, \vec{x'}} & = \min (m(\alpha, M_1(\vec{x})), m(\alpha, M_1(\vec{x'}))) + \min (m(\beta, M_1(\vec{x})), m(\beta, M_1(\vec{x'}))) \\
        & = m(\alpha, M_1(\vec{x})) + m(\beta, M_1(\vec{x})), \text{ and }\\
        N^{nm}_{\vec{x}, \vec{x'}} & = \min (m(\alpha, M_1(\vec{x})) + m(\beta, M_1(\vec{x})), m(\alpha, M_1(\vec{x'})) + m(\beta, M_1(\vec{x'})))\\
        & = m(\alpha, M_1(\vec{x})) + m(\beta, M_1(\vec{x})).
    \end{align*}
    Therefore, $N^m_{\vec{x}, \vec{x'}} = N^{nm}_{\vec{x}, \vec{x'}}$ as claimed. 


    Case 2B: Assume next that the words $\vec{x}$ and $\vec{x'}$ do not satisfy the conditions of 2A, that is, $m(\alpha, M_1(\vec{x})) \leq m(\alpha, M_1(\vec{x'}))$ and $m(\beta, M_1(\vec{x})) > m(\beta, M_1(\vec{x'}))$, or $m(\alpha, M_1(\vec{x})) > m(\alpha, M_1(\vec{x'}))$ and $m(\beta, M_1(\vec{x})) \leq m(\beta, M_1(\vec{x'}))$. Let $i = h + 1$ be the first coordinate where $\vec{x}$ and $\vec{x'}$ differ, that is, $x_i \neq x_i'$. Similarly, let $n - k$ be the last coordinate where $\vec{x}$ and $\vec{x'}$ differ. Hence, the first $h$ symbols and the last $k$ are equal in both $\vec{x}$ and $\vec{x'}$. Notice that if $h + k = n - 1$, then the words differ in exactly one coordinate position and the previously discussed first and last coordinate positions coincide. 
    Let $x_{h + 1} = a$, $x_{h + 1}' = b$, $x_{n - k} = c$ and $x_{n - k}' = d$, where $a \neq b$ and $c \neq d$. 
    For illustrative purposes, we can represent $\vec{x}$ and $\vec{x'}$ as follows:  

{\centering
    \begin{tabular}{c c c c c c c c c c c}
        $\vec{x}$ & $ =$ & $x_1$ & $\dots$ & $x_h$ & $a$ & $\dots$ & $c$ & $x_{n - k + 1}$ & \dots & $x_n$ \\
        $\vec{x'}$ & $ =$ & $x_1$ & $\dots$ & $x_h$ & $b$ & $\dots$ & $d$ & $x_{n - k + 1}$ & $\dots$ & $x_n$\\
    \end{tabular}
    \par
    }
    \noindent since $x_i = x_i'$ for $i \leq h$ and $i \geq n - k + 1$ (when such indices exist, that is, $h > 0$ or $k > 0$).

    Assume without loss of generality that $a\ne x_h$ whenever $h>0$. Indeed, since $a\ne b$, at least one of $a$ and $b$ is different from $x_h$, and we may choose $a$ to be that one. For a word to belong to $A$, a symbol has to be inserted either anywhere before $a$ in $\vec{x}$ or anywhere before $b$ in $\vec{x'}$ (but not to both) since $a \neq b$:
    \begin{itemize}
        \item If a symbol is inserted anywhere before $b$ in $\vec{x'}$, then the insertion has to be $a$ just before $b$ (as $x_h \neq a$ or $h=0$) and a symbol has to be inserted to $\vec{x}$ after $c$ (as $c \neq d$). Hence, the resulting word $\vec{w}_1 = x_1 \cdots x_h a b x'_{h+2} \cdots x'_{n-k-1} d x_{n - k + 1} \cdots x_n$.
        \item If a symbol is inserted anywhere before $a$ in $\vec{x}$, then it is $b$ and a symbol has to be inserted to $\vec{x'}$ after $d$ (as $c \neq d$). Hence, the resulting word $\vec{w}_2 = \vec{x'}_{[1,h+1]}\vec{x}_{[h+1,n]} = x_1 \cdots x_h b a x_{h+2} \cdots x_{n-k-1} c x_{n - k + 1} \cdots x_n$.
    \end{itemize}
    Clearly, $\vec{w}_1 \neq \vec{w}_2$ as their ($h+1$)th symbols $a$ and $b$ disagree. By the previous discussion, they are the only possibilities for the words in $\mathrm{set}(A)$ and both of them have to belong to $A$, as $|\mathrm{set}(A)| = 2$. Hence, we may choose $\alpha = \vec{w}_1$ and $\beta = \vec{w}_2$. Since $\alpha$ can be obtained from $\vec{x'}$ only by inserting the symbol $a$ right before $b$, we have $m(\alpha, M_1(\vec{x'})) = 1$. If $m(\alpha, M_1(\vec{x})) = 1$, then $m(\alpha, M_1(\vec{x})) = m(\alpha, M_1(\vec{x'}))$, and the conditions of Case~2A are met which contradicts the assumptions of Case~2B. Hence, we can assume that $m(\alpha, M_1(\vec{x})) > 1$. However, for $m(\alpha, M_1(\vec{x})) > 1$ to hold, it is required that $k > 0$ and $d = x_{n - k + 1}$. Hence, since $c\neq d$, we have $c \neq x_{n - k + 1}$. Thus, to obtain $\beta$ from $\vec{x'}$, the symbol $c$ has to be inserted just after $d$. Therefore, we have $m(\beta, M_1(\vec{x'})) = 1$ implying $m(\beta, M_1(\vec{x})) > 1$ (as otherwise a contradiction again follows with the assumptions of Case~2B). Hence, together with the observations $m(\alpha, M_1(\vec{x'})) = 1$ and $m(\alpha, M_1(\vec{x})) > 1$, the conditions of Case~2A are satisfied, which again contradicts the assumptions of Case~2B. Thus, $N^m_{\vec{x}, \vec{x'}} = N^{nm}_{\vec{x}, \vec{x'}}$ as claimed. Notice also that the previous arguments apply even if the words $\vec{x}$ and $\vec{x'}$ differ in exactly one coordinate position, that is, $a$ and $c$ as well as $b$ and $d$ coincide; in that case, we simply replace each occurrence of $c$ and $d$ (in the proof) by $a$ and $b$, respectively.

    Thus, in conclusion, $N^m_{\vec{x}, \vec{x'}} = N^{nm}_{\vec{x}, \vec{x'}}$ for all $\vec{x}\neq \vec{x'}$.
   \end{proof}

\begin{remark}\label{samemodel}
    The previous result reveals the surprising fact that the multiset model and the non-multiset model are closely connected when $t=1$ in the following sense: A transmitted word $\vec{x}\in C \subseteq \mathbb{Z}_q^n $ can be determined using a multiset $Y^m$ in the multiset model if and only if it can be determined by the corresponding  $\mathrm{set}(Y^m)$   in the non-multiset model with $N = |Y^m|$.
\end{remark}

Next we determine all possible extremal pairs in the multiset model, when $t = 1$. 
\begin{theorem} 
\label{thm:extremalmultodd}
    Let $n \geq 1$ be an odd integer and let $t = 1$. The only extremal word pairs in the multiset model are of the form $\vec{x} = a^{\frac{n + 1}{2}} b^{\frac{n - 1}{2}}$ and $\vec{x'} = a^{\frac{n - 1}{2}} b^{\frac{n + 1}{2}}$, where $a, b \in \mathbb{Z}_q$ and  $a\neq b$.
\end{theorem}
\begin{proof}
    
    The proof is by induction on the odd length $n$ of a word. When $n = 1$, there is only the pair $\vec{x} = a$ and $\vec{x'} = b$ for each distinct symbols $a$ and $b$ in $\mathbb{Z}_q$, and the statement is trivially true.

    Let us now assume that the statement is true for words of length $n - 2$. That is, suppose that the only extremal pair of words of length $n - 2$ up to permutation of symbols is $a^{\frac{n - 1}{2}} b^{\frac{n - 3}{2}}$ and $a^{\frac{n - 3}{2}} b^{\frac{n - 1}{2}}$. We wish to determine the extremal pairs of length $n\ge 3$. To begin with, recall from Proposition~\ref{prop:x1=x1xn!=xn} and Corollary~\ref{cor:x1!=x1xn=xn} that if the initial or final letters of $\vec{x}$ and $\vec{x'}$ differ, then $|M_1(\vec{x}) \cap M_1(\vec{x'})| \leq 2 < \frac{n + 3}{2}$ when $n \geq 3$. Hence, $\vec{x}$ and $\vec{x'}$ must have the same initial letters and the same final letters. From Corollary~\ref{cor:exactboundt=1}, we know that the size of the intersection of multiset spheres of any extremal pair of length $n$ is equal to $\frac{n + 3}{2}$. But if $\vec{x}$ and $\vec{x'}$ form an extremal pair, then from Theorem~\ref{thm:doublerecupperbound} we have that $\frac{n + 3}{2} = |M_1(\vec{x}) \cap M_1(\vec{x'})| \leq |M_1(\vec{x}_{[2, n - 1]}) \cap M_1(\vec{x'}_{[2, n - 1]})| + 1 \leq \frac{n + 1}{2} + 1$, which implies that $|M_1(\vec{x}_{[2, n - 1]}) \cap M_1(\vec{x'}_{[2, n - 1]})| = \frac{n + 1}{2}$, that is, $\vec{x}_{[2, n - 1]}$ and $\vec{x'}_{[2, n - 1]}$ form an extremal word pair of length $n - 2$. The induction hypothesis gives us that $\vec{x}_{[2, n - 1]} = a^{\frac{n - 1}{2}} b^{\frac{n - 3}{2}}$ and $\vec{x'}_{[2, n - 1]} = a^{\frac{n - 3}{2}} b^{\frac{n - 1}{2}}$. Hence, the extremal word pairs of length $n$ must be pairs of the form $\vec{x} = fa^{\frac{n - 1}{2}} b^{\frac{n - 3}{2}}g, \vec{x'} = fa^{\frac{n - 3}{2}} b^{\frac{n - 1}{2}}g$.
    Moreover, we have from Theorem~\ref{thm:lowerboundt=1odd} that when $\vec{x} = a^{\frac{n + 1}{2}} b^{\frac{n - 1}{2}}$ and $\vec{x'} = a^{\frac{n - 1}{2}} b^{\frac{n + 1}{2}}$, $|M_1(\vec{x}) \cap M_1(\vec{x'})| = \frac{n + 3}{2}$. Hence, to prove that there are no other extremal pairs than the ones in Theorem~\ref{thm:lowerboundt=1odd}, we can assume that $f\neq a$ or $g\neq b.$ In addition,  if we show that the claim holds for $g \neq b$,
then, due to Corollary~\ref{cor:reverse}, the case $f\neq a$ follows. Hence, it is enough to assume from now on that $g\neq b$.

    We first observe that for a word to be in $M_1(\vec{x}) \cap M_1(\vec{x'})$, it must be obtained by inserting the symbol $b$ in $\vec{x}$ or the symbol $a$ in $\vec{x'}$ in order to have the same number of symbols. In what follows, we first consider the case $n = 3$ separately, before proceeding with the case $n \geq 5$.

    If $n = 3$, then the words that can be obtained by inserting the symbol $b$ in $\vec{x}$ are of the form $bfag, fbag, fabg$ or $fagb$.
    As $\vec{x'} = fbg$,  the words of the form $fbag$ or $fabg$ are clearly always in the intersection. Words of the form $fagb$ are never in the intersection since $g \neq b$, and words of the form $bfag$ are in the intersection if and only if $f = b$, in which case $bfag = fbag$.
    Hence, we need to focus on counting words only of the forms $fabg$ and $fbag$ within the intersection. Observe that $m(fbag, M_1(\vec{x})) = 1$ for $f \neq b$, $m(fbag, M_1(\vec{x})) = 2$ for $f = b$, $m(fbag, M_1(\vec{x'})) = 1$ for $g \neq a$, and $m(fbag, M_1(\vec{x'})) = 2$ for $g = a$. Similarly, $m(fabg, M_1(\vec{x})) = 1$ for $g \neq b$, $m(fabg, M_1(\vec{x'})) = 1$ for $f \neq a$, and $m(fabg, M_1(\vec{x'})) = 2$ for $f = a$. Therefore, in this case $|M_1(\vec{x}) \cap M_1(\vec{x'})| = 3 = \frac{n + 3}{2}$ if and only if $f = b$ and $g = a$. 
    Hence, by taking into account the permutation of symbols $a$ and $b$, we get that when $n = 3$, an extremal pair $\vec{x}$ and $\vec{x'}$ is always of the form $a^{\frac{n + 1}{2}} b^{\frac{n - 1}{2}}$ and $a^{\frac{n - 1}{2}} b^{\frac{n + 1}{2}}$. 
    
    Let $n \geq 5$ and consider the words in the intersection obtained by inserting $b$ to $\vec{x}$ and $a$ to $\vec{x'}$. Notice that the words $\vec{x}$ and $\vec{x'}$ differ only in the $(n+1)/2$-th coordinate. Clearly, we cannot have insertions in both $\vec{x}$ and $\vec{x'}$ before this coordinate or after it. If we insert $b$ before the $(n+1)/2$-th coordinate in $\vec{x}$, then it can only be added just after the $(n-1)/2$-th coordinate in $\vec{x}$ due to the number of symbols $b$ (which is zero or one depending on if $f\neq b$ or not) before the $(n+1)/2$-th coordinate in $\vec{x'}$. 
    We can insert $b$ in any position after the $(n+1)/2$-th coordinate, except at the end of $\vec{x}$, since $g \neq b$. Hence, the number of words from $M_1(\vec{x})$ in $M_1(\vec{x}) \cap M_1(\vec{x'})$ is at most $1 + \frac{n - 1}{2} = \frac{n + 1}{2}$, which is a contradiction. This completes the proof.
\end{proof}

\begin{theorem}
\label{thm:extremalmulteven}
    Let $n \geq 2$ be an even integer and let $t = 1$. The only extremal word pairs in the multiset model are of the forms $(i)\ \vec{x} = a^{\frac{n}{2}} b a^{\frac{n - 2}{2}}$ and $\vec{x'} = a^{\frac{n - 2}{2}} b a^{\frac{n}{2}}$, $(ii)\ \vec{x} = a^{\frac{n}{2}} b^{\frac{n - 2}{2}} i$ and $\vec{x'} = a^{\frac{n - 2}{2}} b^{\frac{n}{2}}i$, or $(iii)\ \vec{x} = i a^{\frac{n}{2}} b^{\frac{n - 2}{2}}$ and $\vec{x'} = i a^{\frac{n - 2}{2}} b^{\frac{n}{2}}$ where $a, b, i \in \mathbb{Z}_q$ and $a \neq b$.
\end{theorem}
\begin{proof}

    The proof is by induction on the even length $n$ of a word. When $n = 2$, we have from Corollary~\ref{cor:exactboundt=1}, that the size of the intersection of multiset spheres of any extremal pair of length $n=2$ is equal to $\frac{2 + 2}{2} = 2$. Using Observation~\ref{obs:t=1upperbound} we can see that every extremal pair will be of the form $(i)\ ab, ba$, $(ii)\ ai, bi$ or $(iii)\ ia, ib$, where $i \in \mathbb{Z}_q$. This proves the base case.

    Let us now assume that the statement is true for words of length $n - 2$. That is, suppose that the only extremal pairs of words of length $n - 2$ are of the form $(i)\ \vec{x} = a^{\frac{n - 2}{2}} b a^{\frac{n - 4}{2}}$ and $\vec{x'} = a^{\frac{n - 4}{2}} b a^{\frac{n - 2}{2}}$, $(ii)\ \vec{x} = a^{\frac{n - 2}{2}} b^{\frac{n - 4}{2}} i$ and $\vec{x'} = a^{\frac{n - 4}{2}} b^{\frac{n - 2}{2}}i$, or $(iii)\ \vec{x} = i a^{\frac{n - 2}{2}} b^{\frac{n - 4}{2}}$ and $\vec{x'} = i a^{\frac{n - 4}{2}} b^{\frac{n - 2}{2}}$ where $a, b, i \in \mathbb{Z}_q, a \neq b$, where $i \in \mathbb{Z}_q$. We wish to determine the extremal pairs of length $n\ge 4$. For notational simplicity, we choose $a = 0$ and $b = 1$ as the other cases go exactly the same way due to the permutations of symbols. To begin with, recall from Proposition~\ref{prop:x1=x1xn!=xn} and Corollary~\ref{cor:x1!=x1xn=xn} that if the initial or final letters of $\vec{x}$ and $\vec{x'}$ differ, then $|M_1(\vec{x}) \cap M_1(\vec{x'})| \leq 2 < \frac{n + 2}{2}$ when $n \geq 4$. Hence, $\vec{x}$ and $\vec{x'}$ must have the same initial letters and the same final letters. From Corollary~\ref{cor:exactboundt=1}, we know that the size of the intersection of multiset spheres of any extremal pair of length $n$ is equal to $\frac{n + 2}{2}$. But if $\vec{x}$ and $\vec{x'}$ form an extremal pair, then 
    from Theorem~\ref{thm:doublerecupperbound} we have that $\frac{n + 2}{2} = |M_1(\vec{x}) \cap M_1(\vec{x'})| \leq |M_1(\vec{x}_{[2, n - 1]}) \cap M_1(\vec{x'}_{[2, n - 1]})| + 1 \leq \frac{n}{2} + 1$, which implies that $|M_1(\vec{x}_{[2, n - 1]}) \cap M_1(\vec{x'}_{[2, n - 1]})| = \frac{n}{2}$, that is, $\vec{x}_{[2, n - 1]}$ and $\vec{x'}_{[2, n - 1]}$ form an extremal word pair of length $n - 2$. The induction hypothesis gives us that the pair $\vec{x}_{[2, n - 1]}$ and $\vec{x'}_{[2, n - 1]}$ must be one of the aforementioned word pairs $(i), (ii)$ or $(iii)$. Hence, the extremal word pair must be of the form $f\vec{x}_{[2, n - 1]}g$ and $f\vec{x'}_{[2, n - 1]}g$.

    First, suppose that $\vec{x}_{[2, n - 1]} = 0^{\frac{n - 2}{2}} 1 0^{\frac{n - 4}{2}}, \vec{x'}_{[2, n - 1]} = 0^{\frac{n - 4}{2}} 1 0^{\frac{n - 2}{2}}$. Hence, the extremal word pairs of length $n$ must be pairs of the form $\vec{x} = f0^{\frac{n - 2}{2}} 1 0^{\frac{n - 4}{2}}g, \vec{x'} = f0^{\frac{n - 4}{2}} 1 0^{\frac{n - 2}{2}}g$. We first observe that for a word to be in $M_1(\vec{x}) \cap M_1(\vec{x'})$, the same symbol must be inserted in both $\vec{x}$ and $\vec{x'}$. Next, we study first the case with $n=4$ and then the case with $n \ge 6.$
     
    Let $n = 4$, and suppose that we insert a symbol $h\in\mathbb{Z}_q$ in both words. Since $\vec{x}=f01g$ and $\vec{x'}=f10g$, the words that can be obtained from $\vec{x}$ will be of the forms $hf01g, fh01g, f0h1g, f01hg$ or $f01gh$. Observe that a word of the form $fh01g$ (resp. $f01hg$) is in the intersection if and only if $h = 1$ (resp. $h = 0$). Moreover, a word of the form $hf01g$ (resp. $f01gh$) is in the intersection if and only if $h = f = 1$ (resp. $g = h = 0$). But $h = f$ implies $hf01g = fh01g$ and $g = h$ implies $f01gh = f01hg$. Words of the form $f0h1g$ are never in the intersection, and hence, we need to focus on counting words only of the forms $fh01g$ and $f01hg$ within the intersection. Moreover, if a word $fh01g$ (resp. $f01hg$) is in the intersection, then necessarily $h=1$ (resp. $h=0$). Consequently, we need to focus on counting words only of the forms $f101g$ and $f010g$ within the intersection.  Observe that $m(f101g, M_1(\vec{x})) = 1$ for $f \neq 1$, $m(f101g, M_1(\vec{x})) = 2$ for $f = 1$, $m(f101g, M_1(\vec{x'})) = 1$ for $g \neq 1$, and $m(f101g, M_1(\vec{x'})) = 2$ for $g = 1$. Similarly, $m(f010g, M_1(\vec{x})) = 1$ for $g \neq 0$, $m(f010g, M_1(\vec{x})) = 2$ for $g = 0$, $m(f010g, M_1(\vec{x'})) = 1$ for $f \neq 0$, and $m(f010g, M_1(\vec{x'})) = 2$ for $f = 0$. Therefore, in this case, $|M_1(\vec{x}) \cap M_1(\vec{x'})| = \frac{n + 2}{2} = 3$ if and only if $f = g = 0$ or $f = g = 1$. In either case, by taking the permutation of symbols $0$ and $1$ into account, we get that $\vec{x}$ and $\vec{x'}$ are of the form $0^{\frac{n}{2}} 1 0^{\frac{n - 2}{2}}$ and $0^{\frac{n - 2}{2}} 1 0^{\frac{n}{2}}$. 
     
     Let $n \ge 6$, and suppose first that we insert a symbol $h \neq 0$ in both words. Observe that the inserted symbol must be before the symbol $1$ in $\vec{x}$ and after the symbol $1$ in $\vec{x'}$ so that the number of $0$s before and after the middle symbol $1$ is balanced. Clearly, we must have that $h = 1$ and the word is of the form $f0^{\frac{n - 4}{2}} 1 0 1 0^{\frac{n - 4}{2}}g$. Furthermore, there is only one way of obtaining this word from both $\vec{x}$ and $\vec{x'}$. Suppose next that we insert the symbol $h = 0$ in both words. Observe that the inserted symbol must be after the symbol $1$ in $\vec{x}$ and before the symbol $1$ in $\vec{x'}$ so that the number of $0$s before and after the middle symbol $1$ is balanced. Obtained words are hence of the form $f0^{\frac{n - 2}{2}} 1 0^{\frac{n - 2}{2}} g$. Observe that $m(f0^{\frac{n - 2}{2}} 1 0^{\frac{n - 2}{2}} g, M_1(\vec{x})) = \frac{n - 2}{2}$ if $g \neq 0$ and $m(f0^{\frac{n - 2}{2}} 1 0^{\frac{n - 2}{2}} g, M_1(\vec{x})) = \frac{n}{2}$, if $g = 0$. Similarly, $m(f0^{\frac{n - 2}{2}} 1 0^{\frac{n - 2}{2}} g, M_1(\vec{x'})) = \frac{n - 2}{2}$ if $f \neq 0$ and $m(f0^{\frac{n - 2}{2}} 1 0^{\frac{n - 2}{2}} g, M_1(\vec{x'})) = \frac{n}{2}$, if $f = 0$. Therefore, in this case, $|M_1(\vec{x}) \cap M_1(\vec{x'})| = 1 + \frac{n}{2}$ if and only if $f = g = 0$ and hence, $\vec{x}$ and $\vec{x'}$ are of the forms $0^{\frac{n}{2}} 1 0^{\frac{n - 2}{2}}$ and $0^{\frac{n - 2}{2}} 1 0^{\frac{n}{2}}$ as required.

    Next, suppose that $\vec{x}_{[2, n - 1]} = 0^{\frac{n - 2}{2}} 1^{\frac{n - 4}{2}} i$ and $\vec{x'}_{[2, n - 1]} = 0^{\frac{n - 4}{2}} 1^{\frac{n - 2}{2}}i$. Hence, the extremal word pairs of length $n$ must be pairs of the form $\vec{x} = f0^{\frac{n - 2}{2}} 1^{\frac{n - 4}{2}} ig, \vec{x'} = f0^{\frac{n - 4}{2}} 1^{\frac{n - 2}{2}}ig$. We first observe that for a word to be in $M_1(\vec{x}) \cap M_1(\vec{x'})$, it must be obtained by inserting $1$ in $\vec{x}$ and the symbol $0$ in $\vec{x'}$. 

    If $n = 4$, then $\vec{x}=f0ig$ and $\vec{x'}=f1ig$ and the words that can be obtained by inserting the symbol $1$ in $\vec{x}$ are of the form $1f0ig, f10ig, f01ig, f0i1g$, or $f0ig1$. Clearly, words of the form $f10ig$ or $f01ig$ are always in the intersection. On the other hand, a word of the form $1f0ig$ (resp. $f0i1g$) is in the intersection if and only if $f = 1$ (resp.  $i = 1$), and a word of the form $f0ig1$ is in the intersection if and only if $i = g = 1$. But $f = 1$ implies $1f0ig = f10ig$, $i = 1$ implies $f0i1g = f01ig$, and $i = g = 1$ implies $f0ig1 = f01ig$. Hence, we need to focus on counting words only of the forms $f10ig$ and $f01ig$ within the intersection. Observe that $m(f10ig, M_1(\vec{x})) = 1$ for $f \neq 1$, $m(f10ig, M_1(\vec{x})) = 2$ for $f = 1$, $m(f10ig, M_1(\vec{x'})) = 1$ for $i \neq 0$, and $m(f10ig, M_1(\vec{x'})) \geq 2$ for $i = 0$. Similarly, $m(f01ig, M_1(\vec{x})) = 1$ for $i \neq 1$, $m(f01ig, M_1(\vec{x})) \geq 2$ for $i = 1$, $m(f01ig, M_1(\vec{x'})) = 1$ for $f \neq 0$, and $m(f01ig, M_1(\vec{x'})) = 2$ for $f = 0$. Therefore, in this case, $|M_1(\vec{x}) \cap M_1(\vec{x'})| = \frac{n + 2}{2} = 3$ if $f = 1$ and $ i = 0$, or if $f = 0$ and $i = 1$. In either case, by taking the permutation of symbols $0$ and $1$ into account, we get that $\vec{x}$ and $\vec{x'}$ are of the forms $0^{\frac{n}{2}} 1^{\frac{n - 2}{2}} g$ and $0^{\frac{n - 2}{2}} 1^{\frac{n}{2}}g$.

    Let $n \ge 6 $. Notice that $\vec{x}$ and $\vec{x'}$ differ only in the $\frac{n}{2}$-th coordinate and we must insert the symbol $1$ to $\vec{x}$ and $0$ to $\vec{x'}$. Clearly, the insertion must be done before  the differing coordinate in $\vec{x}$ and after  it in $\vec{x'}$ or the other way around. The resulting word in the first case is of the form $f0^{\frac{n - 4}{2}}1 0 1^{\frac{n - 4}{2}} ig$ and clearly there is only one way to obtain the word, that is, $m(f0^{\frac{n - 4}{2}}1 0 1^{\frac{n - 4}{2}} ig, M_1(\vec{x})) = m(f0^{\frac{n - 4}{2}}1 0 1^{\frac{n - 4}{2}} ig, M_1(\vec{x'})) = 1$. The resulting word of the second case is of  the form $f0^{\frac{n - 2}{2}} 1^{\frac{n - 2}{2}} ig$.
     Observe that $m(f0^{\frac{n - 2}{2}}1^{\frac{n - 2}{2}} ig, M_1(\vec{x})) = \frac{n - 2}{2}$ if $i \neq 1$ and $m(f0^{\frac{n - 2}{2}}1^{\frac{n - 2}{2}} ig, M_1(\vec{x})) \geq \frac{n}{2}$ if $i = 1$. Similarly, $m(f0^{\frac{n - 2}{2}}1^{\frac{n - 2}{2}} ig, M_1(\vec{x'})) = \frac{n - 2}{2}$ if $f \neq 0$ and $m(f0^{\frac{n - 2}{2}}1^{\frac{n - 2}{2}} ig, M_1(\vec{x'})) = \frac{n}{2}$ if $f = 0$. Therefore, in this case, $|M_1(\vec{x}) \cap M_1(\vec{x'})| = \frac{n}{2} + 1$ only if $f = 0$ and $i = 1$, and hence $\vec{x}$ and $\vec{x'}$ are of the forms $0^{\frac{n}{2}} 1^{\frac{n - 2}{2}} g$ and $0^{\frac{n - 2}{2}} 1^{\frac{n}{2}}g$ as required.

    Finally, suppose that $\vec{x}_{[2, n - 1]} = i 0^{\frac{n - 2}{2}} 1^{\frac{n - 4}{2}}$ and $\vec{x'}_{[2, n - 1]} = i 0^{\frac{n - 4}{2}} 1^{\frac{n - 2}{2}}$. Observe that $\overleftarrow{i 0^{\frac{n - 2}{2}} 1^{\frac{n - 4}{2}}} = 1^{\frac{n - 4}{2}} 0^{\frac{n - 2}{2}}i$ and $\overleftarrow{i 0^{\frac{n - 4}{2}} 1^{\frac{n - 2}{2}}} = 1^{\frac{n - 2}{2}} 0^{\frac{n - 4}{2}} i$. Hence, by Corollary~\ref{cor:reverse}, this case can be proved similarly to the previous one. Hence, we have determined all the extremal word pairs of length $n$. This completes the induction step and proves the theorem.
\end{proof}

As a direct consequence of Theorem~\ref{sameproblem}, we can determine the extremal pairs in the non-multiset case from the characterization obtained in Theorems~\ref{thm:extremalmultodd} and \ref{thm:extremalmulteven} for the multiset model.
\begin{corollary}
    Let $t = 1$.
    \begin{enumerate}
        \item If $n \geq 1$ is an odd integer, then the only extremal word pairs in the non-multiset model are of the form $\vec{x} = a^{\frac{n + 1}{2}} b^{\frac{n - 1}{2}}$ and $\vec{x'} = a^{\frac{n - 1}{2}} b^{\frac{n + 1}{2}}$, where $a, b \in \mathbb{Z}_q$ and  $a\neq b$.
        \item If $n \geq 2$ is an even integer, then the only extremal word pairs in the non-multiset model are $(i)\ \vec{x} = a^{\frac{n}{2}} b a^{\frac{n - 2}{2}}, \vec{x'} = a^{\frac{n - 2}{2}} b a^{\frac{n}{2}}$, $(ii)\ \vec{x} = a^{\frac{n}{2}} b^{\frac{n - 2}{2}} i, \vec{x'} = a^{\frac{n - 2}{2}} b^{\frac{n}{2}}i$, and $(iii)\ \vec{x} = i a^{\frac{n}{2}} b^{\frac{n - 2}{2}}, \vec{x'} = i a^{\frac{n - 2}{2}} b^{\frac{n}{2}}$ where $a, b, i \in \mathbb{Z}_q$ and $a \neq b$.
    \end{enumerate}
\end{corollary}

\begin{remark}\label{nmextremal}
    We could also characterize the extremal pairs by first determining them in the non-multiset case using the ideas in the proof of Theorem~\ref{thm:nonmultt=1upperbd} and then obtaining the result for the multiset model due to Theorem~\ref{sameproblem}. However, the proof would be somewhat longer.
\end{remark}

\section{A Bound on \texorpdfstring{$N_q^m(C;n,1)$}{Nq{m}(C;n,1)}}
\label{sec:N_q^m(C;n,1)}
In this section, we consider the parameter $N_q^m(C;n,1)$, which is equal to $N_q^{nm}(C;n,1)$ by Theorem~\ref{sameproblem}, for some codes that are \emph{proper subsets} of $\mathbb{Z}_q^n$.
 With the aid of the bound \eqref{+1Bound} in Theorem~\ref{thm:doublerecupperbound}, we can extend a code $C\subseteq \mathbb{Z}_q^n$ with small $N_q^m(C;n,1)$ to a code $C'\subseteq \mathbb{Z}_q^{n+2}$ with small $N_q^m(C';n+2,1)$.

 For a code $C\subseteq \mathbb{Z}_q^n$ and $a,b\in\mathbb{Z}_q$, we denote $aCb=\{a\vec{c}b\mid \vec{c}\in C\}\subseteq \mathbb{Z}_q^{n+2}$.
\begin{theorem}\label{constr}
Let $C\subseteq \mathbb{Z}_q^n$ be a code. There exists a code  $C'\subseteq \mathbb{Z}_q^{n+2}$ 
with
\begin{equation}\label{constN}
    N_q^m(C';n+2,1)\le \max (N_q^m(C;n,1) +1,2)
\end{equation}
and \begin{equation}\label{constsize}
    |C'|\ge q^2|C|.
\end{equation}
\end{theorem}

\begin{proof}
    Let  $C\subseteq \mathbb{Z}_q^n$ be a code 
    and
    \begin{equation}\label{subsetconstr}
      C'=\bigcup_{ab\in \mathbb{Z}_q^2} aCb \subseteq \mathbb{Z}_q^{n+2}.  
    \end{equation}
       Clearly, $|C'|=q^2|C|.$
    Let $\vec{x}\in a_1Cb_1$ and $\vec{x'}\in a_2Cb_2$ be two distinct codewords of $C'$, where $a_1b_1,a_2b_2\in \mathbb{Z}_q^2 $. If  $a_1\neq a_2$ or $b_1\neq b_2$ (or both), then, by Proposition~\ref{prop:x1=x1xn!=xn} and Corollary~\ref{cor:x1!=x1xn=xn}, we know that $|M_1(\vec{x})\cap M_1(\vec{x'})|\le 2$. If $a_1=a_2$ and $b_1=b_2$, then, by \eqref{+1Bound}, we obtain
    $|M_1(\vec{x})\cap M_1(\vec{x'})|\le N_q^m(C;n,1)+1$. The assertion now follows.
\end{proof}

In general, neither of the bounds \eqref{constN} and \eqref{constsize} in the previous result can be improved. Namely, if $C= \mathbb{Z}_q^1$, then we know  by Corollary~\ref{cor:exactboundt=1}, that $N_q^m(C;1,1)=2$. The code $C'$ in \eqref{subsetconstr} is a code of size $q^2\cdot q=q^3$ in $\mathbb{Z}_q^3$ with $N_q^m(C';3,1)\le 2+1=3$. Moreover, by Corollary~\ref{cor:exactboundt=1}, we know that the bound is attained, that is, $N_q^m(C';3,1)=3$ and thus \eqref{constN} is tight. Furthermore, there cannot exist a larger code in $\mathbb{Z}_q^3$ than $C'$ implying that the bound \eqref{constsize} is also tight.

If $N_q^m(C;n,1)\ge 2$, then \eqref{constN} gives  $N_q^m(C';n+2,1)\le N_q^m(C;n,1) +1$. Notice that it is possible to have $N_q^m(C';n+2,1)= N_q^m(C;n,1)$ in the previous theorem, that is, the number of channels can even remain the same for some well-chosen $C$ and $C'$. This happens if  \eqref{+0Bound} holds with equality (so we can use it in the proof of Theorem~\ref{constr} instead of \eqref{+1Bound}). Indeed, we can deduce the word pairs such that \eqref{+0Bound} holds with equality as explained in Remark~\ref{RmkCharacterization} in Appendix. Using the approach of the remark, it is straightforward to check, for example, that the code $C=\{0010,0110\}\subseteq \mathbb{Z}_2^4$ with $N_2^m(C;4,1)=3$ gives the code $C'\subseteq \mathbb{Z}_2^6$ with $N_2^m(C';6,1)=3$ as shown in Example~\ref{+0Conditions}.

Let us next consider codes with $N_q^m(C;n,1)\le 1$.

\begin{corollary}\label{smallNcodes}
    Suppose that $C\subseteq \mathbb{Z}_q^n$ is a code such that
    \begin{equation}\label{codecond}
        |M_1(\vec{x})\cap M_1(\vec{x'})|\le 1
    \end{equation}
     for all $\vec{x},\vec{x'}\in C$ and $\vec{x}\neq \vec{x'}$. There exists a code $C'\subseteq \mathbb{Z}_q^{n+2}$ with $N_q^m(C';n+2,1)\le 2$ of cardinality $q^2|C|.$ 
\end{corollary}

\begin{proof}
    From \eqref{codecond}, we get $N_q^m(C;n,1)\le 1$ and, hence, by \eqref{constN}, we have $|M_1(\vec{z})\cap M_1(\vec{z'})|\le \max(1+1,2)=2$ for all distinct $\vec{z},\vec{z'}\in C'\subseteq \mathbb{Z}_q^{n+2}$, where $C'$ is defined as in \eqref{subsetconstr}.
\end{proof}

It is easy to verify that the code $C=\{000,011,110\}\subseteq \mathbb{Z}_2^3$ satisfies the condition \eqref{codecond}. Hence, we obtain a code $C'$ of 12 codewords in $\mathbb{Z}_2^5$ with $N_q^m(C';5,1)\le 2$. In fact, the upper bound in Corollary~\ref{smallNcodes} is tight, that is,  $N_q^m(C';5,1)=2$, since $M_1(00000)\cap M_1(00001)=\{000010,000001\}$ where $00000,00001\in C'$. 

Recall that $\mathrm{set}(M_1(\vec{x}))=I_1(\vec{x}).$ In the literature (see, for example, Tenengolts \cite{Tenengolts}), there are codes $C \subseteq  \mathbb{Z}_q^n$  such that $I_1(\vec{x})\cap I_1(\vec{x'})=\emptyset$ for all $\vec{x},\vec{x'}\in C$, $\vec{x}\neq \vec{x'}$. Hence, such codes satisfy the condition \eqref{codecond}. Consequently, these codes $C \subseteq \mathbb{Z}_q^n$ imply codes of cardinality $q^2|C|$ in $\mathbb{Z}_q^{n+2}$ such that it is enough to have three channels with different insertion errors to unambiguously deduce the transmitted codeword in both the multiset and non-multiset models.

\section{Some General Bounds for \texorpdfstring{$t\ge 1$}{t geq 1}}
\label{sec:generalbds}
In this section, we give some bounds on $N_q^m(n, t)$ 
 for arbitrary $n, q$ and $t$. To begin with, the following result gives a lower bound on the value $N_q^{m}(n, t)$ by estimating
the size of the intersection of the insertion spheres of a particular set of words.

\begin{theorem}\label{GeneralLBMultiset}
For $n\geq 3$,
we have 
$$N_q^m(n, t) \geq q^{t-1}\left(\left(\sum_{i=0}^{\lceil\frac{n}{2}\rceil-1}\binom{n+t-2-2i}{t-1}\right)-\binom{\lceil\frac{n}{2}\rceil+t-3}{t-1}\right) \geq q^{t-1}\left(\frac{n+t-1}{2n+t-1}\binom{n+t-1}{t}-\binom{\lceil\frac{n}{2}\rceil+t-3}{t-1}\right).$$
\end{theorem}
\begin{proof}
    Let $\vec{x} = 0^{\lfloor \frac{n}{2} \rfloor} 1^{\lceil \frac{n}{2} \rceil}$ and $\vec{x'} = 0^{\lfloor \frac{n}{2} \rfloor + 1} 1^{\lceil \frac{n}{2} \rceil - 1}$. We will show  that $ |M_t(\vec{x}) \cap M_t(\vec{x'})| \geq q^{t-1}\left(\left(\sum_{i=0}^{\lceil\frac{n}{2}\rceil-1}\binom{n+t-2-2i}{t-1}\right)-\binom{\lceil\frac{n}{2}\rceil+t-3}{t-1}\right)$
    which will prove the first inequality of the theorem. Every insertion vector which we consider in this proof, has weight exactly $t$. In the following, we first construct two sets of insertion vectors $W_1$ and $W_2$ which will be applied to $\vec{x}$, and count the sizes of the sets. Based on $W_1$ and $W_2$, we construct sets $W_1'$ and $W_2'$ such that for each $w\in W_i$ $(i=1,2)$ we obtain a unique insertion vector $w'\in W_i'$ $(i=1,2)$. Moreover,  applying $w$ to $\vec{x}$ leads to the same word as applying $w'$ to $\vec{x'}$, that is, the obtained word is in the multiset $M_t(\vec{x}) \cap M_t(\vec{x'})$, and thus, it contributes to the lower bound.
    
    We begin by constructing a set $W_1$ of insertion vectors acting on $\vec{x}$. For all $w=(\vec{w_0}, \dots, \vec{w_n})\in W_1$
    we require that $\vec{w_n}=\varepsilon$. Additionally, there is no insertion vector $w\in W_1$ with $\vec{w_i}=\varepsilon$ for each $i\in[0,\lfloor\frac{n}{2}\rfloor]$. A vector $w$ with $\vec{w_0}\neq\varepsilon$ belongs to $W_1$ if and only if $\vec{w_0}=0\vec{v}$ for some $\vec{v}\neq\varepsilon$ where the weight of $\vec{v}$ is at most $t-1$ and the words $\vec{w_i}$ with $0<i<n$ are freely chosen as long as $w$ has weight $t$. This choice gives $M_q(n-1,t-1)-M_q(n-2,t-1)$ insertion vectors in $W_1$, since $M_q(n-1,t-1)$ corresponds to placing $t-1$ symbols to positions other than the last one and from this we remove the number $M_q(n-2,t-1)$ corresponding to insertion vectors which also have $\vec{v}=\varepsilon$. 
    Similarly, for each $j \in [1, \lfloor n/2 \rfloor]$, a vector $w$ with $\vec{w_i}=\varepsilon$ for every $i\in[0,j-1]$ and $\vec{w_{j}}\neq\varepsilon$ belongs to $W_1$ if and only if $\vec{w_{j}}=0\vec{v}$ for some $\vec{v}\neq\varepsilon$ of weight at most $t-1$ (and the words $\vec{w_i}$ with $j<i<n$ are freely chosen as long as $w$ has weight $t$). In other words, we have \[
     W_1 = \left\{(\varepsilon, \ldots, \varepsilon, 0\vec{v}, \vec{w_{j+1}}, \ldots, \vec{w_{n-1}}, \varepsilon) \in W \mid \vec{v} \neq \varepsilon \text{ and } j \in \left[0, \left\lfloor \frac{n}{2} \right\rfloor \right] \right\} \text,
    \]
     where $W$ is the set of all insertion vectors of weight $t$. This leads to $M_q(n-j-1,t-1)-M_q(n-j-2,t-1)$ insertion vectors in $W_1$ for each $j$, where $M_q(n-j-1,t-1)$ corresponds to insertion vectors with $j$ empty words at beginning and one at the end of the insertion vector from which we remove $M_q(n-j-2,t-1)$ such insertion vectors which also have $\vec{v}=\varepsilon$.   
    In total, with these restrictions, we have by Equation~\eqref{eq:multisetspheresize} 
    \begin{align*}
        |W_1|=&\sum_{j=0}^{\lfloor\frac{n}{2}\rfloor}\left(M_q(n-j-1,t-1)-M_q(n-j-2,t-1)\right)\\
        =&q^{t-1}\sum_{j=0}^{\lfloor\frac{n}{2}\rfloor}\left(\binom{n+t-2-j}{t-1}-\binom{n+t-3-j}{t-1}\right)\\
        =&q^{t-1}\left(\binom{n+t-2}{t-1}-\binom{\lceil\frac{n}{2}\rceil+t-3}{t-1}\right) \text,
    \end{align*}
    where the last equality follows due to the telescopic nature of the sum. Note that for $t=1$ there are no insertion vectors in $W_1$. However, the set $W_2$ discussed in the following paragraph is non-empty even in this case.

    Let us then consider the set $W_2$ of insertion vectors. The set $W_2$ is constructed somewhat similarly to the set $W_1$; namely, the first non-empty insertion $\vec{w_j}=0\vec{v}$ has $\vec{v} = \varepsilon$ instead of $\vec{v} \neq \varepsilon$ (as in $W_1$) and we also require (unlike in $W_1$) that $\vec{w_{n-j-1}} = \vec{w_{n-j}} = \cdots = \vec{w_{n}} = \varepsilon$. 
    For $w = (\vec{w_0},\dots,\vec{w_n}) \in W_2$, we have the following restrictions on it: we do not have $\vec{w_i}=\varepsilon$ for all $i\in[0,\lceil\frac{n}{2}\rceil-2]$ for any $w$ in $W_2$. Notice that since $n\geq3$, we have $\lceil\frac{n}{2}\rceil-2\geq0$. The second restriction is that if $\vec{w_i}=\varepsilon$ for each $i<j$ with some $j\leq \lceil\frac{n}{2}\rceil-2$ and $\vec{w_j}\neq\varepsilon$ for $w\in W_2$, then $\vec{w_j}=0$, and $\vec{w_n}=\vec{w_{n-1}}=\varepsilon$ together with $\vec{w_{n-i-2}}=\varepsilon$ for each $i<j$. The other words $\vec{w_i}$ of $w$ are chosen freely as long as the weight of $w$ is $t$. 
    Notice that $n-i-2>n-(\lceil\frac{n}{2}\rceil-2)-2=\lfloor\frac{n}{2}\rfloor>\lceil\frac{n}{2}\rceil-2$ and thus, we do not define any word $\vec{w_i}$ in $w\in W_2$ simultaneously as an empty word and also as a non-empty word. In other words, we have \[
     W_2 = \left\{(\varepsilon, \ldots, \varepsilon, 0, \vec{w_{j+1}}, \ldots, \vec{w_{n-j-2}}, \varepsilon, \ldots, \varepsilon) \in W \mid  j \in \left[0, \left\lceil \frac{n}{2} \right\rceil - 2 \right] \right\} \text,
     \]
     where $W$ is the set of all insertion vectors of weight $t$. In particular, for each $j$, we add $M_q(n-2j-3,t-1)$ insertion vectors to $W_2$. Indeed, as the first non-empty $\vec{w_j}$ in the insertion vector is fixed as $0$, we may consider $t-1$ insertions. Moreover, there are $(j+1)+(j+2)=2j+3$ positions in which we cannot insert other symbols. 
    Hence,  
      \begin{align*}
        |W_2|=&\sum_{j=0}^{\lceil\frac{n}{2}\rceil-2}M_q(n-2j-3,t-1)\\
        =&\sum_{j=0}^{\lceil\frac{n}{2}\rceil-2}q^{t-1}\binom{n+t-4-2j}{t-1}.
        \end{align*}
 We have $W_1\cap W_2=\emptyset$ since for the first non-empty insertion $0\vec{v}$, we have $\vec{v}\neq\varepsilon$ in $W_1$ and $\vec{v}=\varepsilon$ in $W_2$. 
 Hence, the combined size of $W_1$ and $W_2$ is 
 \begin{align}\label{w1+w2}
        |W_1|+|W_2| =&  q^{t-1}\left(\binom{n+t-2}{t-1}-\binom{\lceil\frac{n}{2}\rceil+t-3}{t-1}\right) +    \sum_{i=0}^{\lceil\frac{n}{2}\rceil-2}q^{t-1}\binom{n+t-4-2i}{t-1}\nonumber \\
        =&q^{t-1}\left(\left(\sum_{i=0}^{\lceil\frac{n}{2}\rceil-1}\binom{n+t-2-2i}{t-1}\right)-\binom{\lceil\frac{n}{2}\rceil+t-3}{t-1}\right).
    \end{align}

    Let us next construct the sets of insertion vectors $W_1'$ and $W_2'$ which are applied to $\vec{x'}$. Consider first the set $W_1'$. We construct the vectors $w'\in W_1'$ using $W_1$ in the following way: For any $w\in W_1$ with $w=(\vec{w_0},\dots,\vec{w_n})$ where each $\vec{w_i}=\varepsilon$ for some $i<j$ and $\vec{w_j}=0\vec{v}$ for some $\vec{v}\neq\varepsilon$, we construct $w'$ where $\vec{w'_i}=\vec{w_i}=\varepsilon$ for every $i<j$, $\vec{w'_j}=\varepsilon$, $\vec{w'_{j+1}}=\vec{v}$ and $\vec{w_{i'}'}=\vec{w_{i'-1}}$ for $j+1<i'<n$, and $\vec{w_n'}=\vec{w_{n-1}}1$.  In other words, we have 
    \[W'_1 = \left\{(\varepsilon, \ldots, \varepsilon, \vec{w'_{j+1}} = \vec{v}, \vec{w'_{j+2}} = \vec{w_{j+1}}, \ldots, \vec{w'_{n-1}} = \vec{w_{n-2}}, \vec{w'_{n}} = \vec{w_{n-1}}1) \mid \vec{v} \neq \varepsilon \text{ and } j \in \left[0, \left\lfloor \frac{n}{2} \right\rfloor \right] \right\} \text.
     \]
     We observe that the insertion vector $w \in W_1$ applied to $\vec{x}$ and the corresponding $w' \in W'_1$ applied to $\vec{x'}$ lead to the same word 
     \[
     0^{j+1}\vec{v}0\vec{w_{j+1}}0 \cdots 0\vec{w_{\lfloor n/2 \rfloor}}1\vec{w_{\lfloor n/2 \rfloor+1}}1 \cdots 1 \vec{w_{n-1}}1    
     \]
     in the multiset $M_t(\vec{x}) \cap M_t(\vec{x'})$.   
    Furthermore, each insertion vector in $W_1'$ is unique and has weight $t$. Hence, $|W_1|=|W_1'|$ and we obtain each word with same multiplicity with insertion vectors of $W_1$ and $W_1'$.

    Let us then construct set $W_2'$. Consider again $w\in W_2$ with $w=(\vec{w_0},\dots,\vec{w_n})$. Suppose that $\vec{w_i}=\varepsilon$ for each $i<j $ for some $j\leq \lceil\frac{n}{2}\rceil-2$ and $\vec{w_j}=0$. By the definition of $W_2$, we have $\vec{w_n}=\vec{w_{n-1}}=\vec{w_{n-i-2}}=\varepsilon$ for each $i<j$. We construct for such $w$ the insertion vector $w'\in W_2'$ with $w'=(\vec{w_0'},\dots,\vec{w_n'})$ where $ \vec{w_i'}=\varepsilon$ for each $i\leq j+1$ and $i\geq n-j$, and $\vec{w_i'}=\vec{w_{i-1}}$ for each $j+2\leq i\leq n-2-j$, and $\vec{w'_{n-1-j}}=\vec{w_{n-2-j}}1$. Note that we have $j+1<n-1-j$ since  $j\leq\lceil\frac{n}{2}\rceil-2$. Hence, these conditions do not overlap. In other words, we have \[
     W'_2 = \left\{ \left(\varepsilon, \ldots, \varepsilon, \vec{w'_{j+2}} = \vec{w_{j+1}}, \ldots, \vec{w'_{n-j-2}} = \vec{w_{n-j-3}}, \vec{w'_{n-j-1}} = \vec{w_{n-j-2}}1, \varepsilon, \ldots, \varepsilon \right)  
     \mid  j \in \left[0, \left\lceil \frac{n}{2} \right\rceil - 2 \right] \right\} \text.
     \] 
     We observe that the insertion vector $w \in W_2$ applied to $\vec{x}$ and the corresponding $w' \in W'_2$ applied to $\vec{x'}$ lead to the same word 
     \[
     0^{j+2}\vec{w_{j+1}}0\vec{w_{j+2}}0 \cdots 0\vec{w_{\lfloor n/2 \rfloor}}1\vec{w_{\lfloor n/2 \rfloor+1}}1 \cdots 1 \vec{w_{n-j-2}}1^{j+2}    
     \]
     in the multiset $M_t(\vec{x}) \cap M_t(\vec{x'})$. Moreover, each insertion vector in $W_2'$  has weight $t$ and is unique, that is, no two distinct insertion vectors of $W_2$ imply the same vector of $W'_2$ (in the definition above). Indeed, if $w,w'\in W_2'$ and $w=w'$, then both begin with at least $j'+2$ `forced' empty words and 
    end with exactly $j'+1$ empty words $\varepsilon$ for some $j'\geq0$ while each other word of the insertion vectors is determined by $\vec{w_i}$ and $\vec{w_i'}$ from the corresponding insertion vectors $w_2$ and $w_2'$ of $W_2$ for $i\in[j'+1,n-j'-2]$. 
    Hence, $w_2=w_2'$, a contradiction.    
    Therefore, also $|W_2|=|W_2'|$ and we again obtain each word with same multiplicity with insertion vectors of $W_2$ and $W_2'$. Moreover, $W_1'\cap W_2'=\emptyset$ as every insertion vector $w$ of $W_1'$ has $\vec{w_n}\neq\varepsilon$ while it is always an empty word for any insertion vector of $W_2'$.
    
    Therefore, the first inequality of the theorem follows from our computation of $|W_1|+|W_2|$.

\smallskip

Let us now consider the latter inequality in the claim of the theorem. For this, we study the lower bound of~\eqref{w1+w2}. First,  for $0\leq i\leq\lfloor\frac{n}{2}\rfloor-1$ (the case $n$ is odd and $i=\lceil\frac{n}{2}\rceil-1$ of \eqref{w1+w2} is considered separately below), we notice  (using $x/y\ge (x+1)/(y+1)$, when $x\ge y>0$, obtained from the well-known mediant inequality)  that
 \begin{align*}
        \binom{n+t-2-2i}{t-1}=&\frac{n+t-1}{2n+t-1}\left(\binom{n+t-2-2i}{t-1}+\frac{n}{n+t-1}\binom{n+t-2-2i}{t-1}\right)\\
        =&\frac{n+t-1}{2n+t-1}\left(\binom{n+t-2-2i}{t-1}+\frac{n}{n+t-1}\frac{n+t-2-2i}{n-1-2i}\binom{n+t-3-2i}{t-1}\right)\\
        \geq& \frac{n+t-1}{2n+t-1}\left(\binom{n+t-2-2i}{t-1}+\frac{n}{n+t-1}\frac{n+t-1}{n}\binom{n+t-3-2i}{t-1}\right)\\
        =&\frac{n+t-1}{2n+t-1}\left(\binom{n+t-2-2i}{t-1}+\binom{n+t-3-2i}{t-1}\right) \text,
    \end{align*} 
    where we especially note that $n-1-2i > 0$ as $i\leq\lfloor\frac{n}{2}\rfloor-1$. Furthermore, the inequality holds also when $n$ is odd and $i=\lceil\frac{n}{2}\rceil-1$. Indeed, then the left hand side equals to 1 and the right hand side to $\frac{n+t-1}{2n+t-1}<1$. 
With the aid of the above inequality, we get the second claim of the theorem from  \eqref{w1+w2} by using the well-known hockey-stick identity (in the last equality below),
\begin{align*}
        |W_1|+|W_2| 
        \geq& q^{t-1}\left(\frac{n+t-1}{2n+t-1}\left(\sum_{i=0}^{\lceil\frac{n}{2}\rceil-1}\binom{n+t-2-2i}{t-1}+\binom{n+t-3-2i}{t-1}\right)-\binom{\lceil\frac{n}{2}\rceil+t-3}{t-1}\right)\\
        =&q^{t-1}\left(\frac{n+t-1}{2n+t-1}\left(\sum_{i=t-1}^{n+t-2}\binom{i}{t-1}\right)-\binom{\lceil\frac{n}{2}\rceil+t-3}{t-1}\right)\\
        = &q^{t-1}\left(\frac{n+t-1}{2n+t-1}\binom{n+t-1}{t}-\binom{\lceil\frac{n}{2}\rceil+t-3}{t-1}\right).\qedhere
    \end{align*}
   \end{proof}

   Notice that for $t=1$ we have  $N_q^m(n, 1) = \lceil \frac{n + 2}{2} \rceil$ by  Corollary~\ref{cor:exactboundt=1}. However, the first general bound above in Theorem~\ref{GeneralLBMultiset} gives $N_q^m(n, 1) \geq \lceil \frac{n}{2} \rceil-1$ for $t=1$, which is only two less than the exact value. 

   Regarding  the bounds on $N_q^{nm}(n, t)$ for the non-multiset case, see the discussion in Conclusion, but, of course, we trivially have $N_q^{nm}(n, t) \geq N_q^m(n, t)$. 
  \medskip

Next, we determine an \emph{upper} bound on the value of $N_q^{m}(n, t)$. To obtain this upper bound, we draw inspiration from Levenshtein's work on insertion errors~\cite{levenshtein2001efficient} and use recursion to determine the elements of an insertion sphere. Unlike in \eqref{eq:doubrec}, we shall focus only on the first position of an output word $\vec{y}$. There are two possibilities to insert symbols to $\vec{x} = x_1x_2 \cdots x_n$: either no insertions occurred before $x_1$ or an insertion took place before $x_1$. If no insertions have taken place before $x_1$, then a total of $t$ insertions take place in $\vec{x}_{[2, n]}$. If an insertion $\alpha$ has taken place before $x_1$, then a total of $t - 1$ insertions take place in $\alpha\vec{x}$ after $\alpha$. We can write this symbolically as follows for $n\ge 2$: 
\begin{equation}
    \label{eq:singrec}
    M_t(\vec{x}) = x_1 M_t(\vec{x}_{[2, n]}) \underset{a \in \mathbb{Z}_q}\bigoplus a M_{t - 1}(\vec{x}).
\end{equation}
Using this, we can give a recursive formula for the intersection of the insertion spheres of two words. Let $\vec{x} = \vec{x}_{[1, n]}$ and $\vec{x'} = \vec{x'}_{[1, n]}$ be distinct words in $\mathbb{Z}_q^n$. Applying (\ref{eq:singrec}) to both words and grouping the resulting words according to their first symbol, we obtain the following two cases:

\noindent \textbf{Case 1:} When $x_1 = x_1'$, we have:
\begin{equation}
\label{eq:singlerecursionx1=x1'}
    \begin{split}
        M_t(\vec{x}) \cap M_t(\vec{x'}) = & \big[(x_1 M_t(\vec{x}_{[2, n]}) \oplus x_1 M_{t - 1}(\vec{x})) \cap (x_1 M_t(\vec{x'}_{[2, n]}) \oplus x_1 M_{t - 1}(\vec{x'}))\big] \\  & \underset{\substack{a \in \mathbb{Z}_q \\a \neq x_1}}\bigoplus \big[a M_{t - 1}(\vec{x}) \cap a M_{t - 1}(\vec{x'})\big].
    \end{split}
\end{equation}
\noindent\textbf{Case 2:} When $x_1 \neq x_1'$, we have:
\begin{equation}
\label{eq:singlerecursionx1!=x1'}
    \begin{split}
        M_t(\vec{x}) \cap M_t(\vec{x'}) = & \big[(x_1 M_t(\vec{x}_{[2, n]}) \oplus x_1 M_{t - 1}(\vec{x})) \cap (x_1 M_{t - 1}(\vec{x'}))\big] \oplus \big[x_1' M_{t - 1}(\vec{x}) \\ & \cap (x_1' M_t(\vec{x'}_{[2, n]}) \oplus x_1' M_{t - 1}(\vec{x'}))\big] \underset{\substack{a \in \mathbb{Z}_q \\a \neq x_1, x_1'}}\bigoplus \big[a M_{t - 1}(\vec{x}) \cap a M_{t - 1}(\vec{x'})\big].
    \end{split}
\end{equation}

The recursive relations (\ref{eq:singlerecursionx1=x1'}) and (\ref{eq:singlerecursionx1!=x1'}) can be used to give us some upper bounds on the value of $N^m_q(n, t)$ depending on the approximations made. The following bound is an improvement of a bound in \cite[Theorem 3.1]{pavanITWInsertion}.

\begin{theorem}
\label{thm:singrecubB}
    If $n\ge 2$ and $t \geq 2$, then $N^m_q(n, t) \leq N^m_q(n - 1, t) + (q - 1) N^m_q(n, t - 1) + 2M_q(n, t - 1)$.
\end{theorem}
\begin{proof}
    Let $\vec{x}, \vec{x'} \in \mathbb{Z}_q^n, \vec{x} \neq \vec{x'}$ be such that $|M_t(\vec{x}) \cap M_t(\vec{x'})| = N^m_q(n, t)$. If $x_1 = x_1'$, then, from Equation (\ref{eq:singlerecursionx1=x1'}), we can take the cardinality of both sides and simplify as follows
    \begin{align*}
        |M_t(\vec{x}) \cap M_t(\vec{x'})| = \ & |(x_1 M_t(\vec{x}_{[2, n]}) \oplus x_1 M_{t - 1}(\vec{x})) \cap (x_1 M_t(\vec{x'}_{[2, n]}) \oplus x_1 M_{t - 1}(\vec{x'}))| \\  & +|\underset{\substack{a \in \mathbb{Z}_q \\a \neq x_1}}\bigoplus a M_{t - 1}(\vec{x}) \cap a M_{t - 1}(\vec{x'})| \\ =\ & |(M_t(\vec{x}_{[2, n]}) \oplus M_{t - 1}(\vec{x})) \cap (M_t(\vec{x'}_{[2, n]}) \oplus M_{t - 1}(\vec{x'}))| \\  & + (q - 1)|M_{t - 1}(\vec{x}) \cap M_{t - 1}(\vec{x'})| \\ \leq\ & |M_t(\vec{x}_{[2, n]}) \cap M_t(\vec{x'}_{[2, n]})| + |M_{t - 1}(\vec{x}) \cap (M_t(\vec{x'}_{[2, n]}) \oplus M_{t - 1}(\vec{x'}))| \\  & + |(M_t(\vec{x}_{[2, n]}) \oplus M_{t - 1}(\vec{x})) \cap M_{t - 1}(\vec{x'})| + (q - 1)|M_{t - 1}(\vec{x}) \cap M_{t - 1}(\vec{x'})|.
\end{align*}
Since $|M_t(\vec{x}_{[2, n]}) \cap M_t(\vec{x'}_{[2, n]})| \leq N^m_q(n - 1, t)$ and $|M_{t - 1}(\vec{x}) \cap M_{t - 1}(\vec{x'})| \leq N^m_q(n, t - 1)$, the above equation can be simplified to get
\begin{equation*}
\begin{split}
     N^m_q(n, t) \leq\ & N^m_q(n - 1, t) + |M_{t - 1}(\vec{x}) \cap (M_t(\vec{x'}_{[2, n]}) \oplus M_{t - 1}(\vec{x'}))| \\& + |(M_t(\vec{x}_{[2, n]}) \oplus M_{t - 1}(\vec{x})) \cap M_{t - 1}(\vec{x'})| + (q - 1) N^m_q(n, t - 1).
\end{split}
\end{equation*}
But $|M_{t - 1}(\vec{x}) \cap (M_t(\vec{x'}_{[2, n]}) \oplus M_{t - 1}(\vec{x'}))| \leq |M_{t - 1}(\vec{x})| = M_q(n, t - 1)$ and $|(M_t(\vec{x}_{[2, n]}) \oplus M_{t - 1}(\vec{x})) \cap M_{t - 1}(\vec{x'})| \leq |M_{t - 1}(\vec{x'})| = M_q(n, t - 1)$. Hence, we get
\begin{equation}
\label{eq:singlerecursionubB1}
    N^m_q(n, t) \leq N^m_q(n - 1, t) + (q - 1) N^m_q(n, t - 1) + 2M_q(n, t - 1).
\end{equation}

On the other hand, if $x_1 \neq x_1'$, then, from Equation (\ref{eq:singlerecursionx1!=x1'}), we can take the cardinality of both sides and simplify as follows
\begin{align*}
        |M_t(\vec{x}) \cap M_t(\vec{x'})| =\ & |(x_1 M_t(\vec{x}_{[2, n]}) \oplus x_1 M_{t - 1}(\vec{x})) \cap x_1 M_{t - 1}(\vec{x'})| + |x_1' M_{t - 1}(\vec{x}) \\ & \cap (x_1' M_t(\vec{x'}_{[2, n]}) \oplus x_1' M_{t - 1}(\vec{x'}))| + |\underset{\substack{a \in \mathbb{Z}_q \\a \neq x_1, x_1'}}\bigoplus [a M_{t - 1}(\vec{x}) \cap a M_{t - 1}(\vec{x'})]| \\ =\ & |(M_t(\vec{x}_{[2, n]}) \oplus M_{t - 1}(\vec{x})) \cap M_{t - 1}(\vec{x'})| + |M_{t - 1}(\vec{x})\\ & \cap (M_t(\vec{x'}_{[2, n]}) \oplus M_{t - 1}(\vec{x'}))| + (q - 2)|M_{t - 1}(\vec{x}) \cap M_{t - 1}(\vec{x'})|.
\end{align*}
Since $|M_{t - 1}(\vec{x}) \cap M_{t - 1}(\vec{x'})| \leq N^m_q(n, t - 1)$, the above equation can be simplified to get
\begin{equation*}
\begin{split}
    N^m_q(n, t) \leq\ & |(M_t(\vec{x}_{[2, n]}) \oplus M_{t - 1}(\vec{x})) \cap M_{t - 1}(\vec{x'})| + |M_{t - 1}(\vec{x}) \cap (M_t(\vec{x'}_{[2, n]}) \oplus M_{t - 1}(\vec{x'}))| \\ & + (q - 2) N^m_q(n, t - 1).
\end{split}
\end{equation*}
But just as before, $|(M_t(\vec{x}_{[2, n]}) \oplus M_{t - 1}(\vec{x})) \cap M_{t - 1}(\vec{x'})| \leq |M_{t - 1}(\vec{x'})| = M_q(n, t - 1)$ and $|M_{t - 1}(\vec{x}) \cap (M_t(\vec{x'}_{[2, n]}) \oplus M_{t - 1}(\vec{x'})) \leq |M_{t - 1}(\vec{x})| = M_q(n, t - 1)$. Hence, we get
\begin{equation}
    \label{eq:singlerecursionubB2}
        N^m_q(n, t) \leq (q - 2) N^m_q(n, t - 1) + 2M_q(n, t - 1).
\end{equation}
Equations (\ref{eq:singlerecursionubB1}) and (\ref{eq:singlerecursionubB2}) together complete the proof.
\end{proof}

\begin{remark}
    There is probably a lot of room for improvement in the bounds provided by Theorems~\ref{thm:singrecubB} and \ref{GeneralLBMultiset}.
    For example, when $q = 2$, $n = 4$ and $t = 4$, we obtain from Equation~\eqref{w1+w2} that $N_2^m(4, 4) \geq 184$, and by substituting the values $N_2^m(3, 4) = 324$ and $N_2^m(4, 3) = 142$ (which can be determined using a computer program) in Equation~\eqref{eq:singlerecursionubB1}, we obtain that $N_2^m(4, 4) \leq 1026$, while the actual value of $N_2^m(4, 4)$ is 640. 
\end{remark}

\section{Conclusion}
\label{sec:conc}
In this paper, we have explored how the multiset and the non-multiset models address some of the deficiencies in the influential Levenshtein's sequence reconstruction problem. We have completely determined the minimum number of channels needed in these models for determining the transmitted word unambiguously in the case of one insertion error for any $q\ge 2$ and $n\ge 1$. 
These results are also obtained when the words are of length one for any $q\ge 2 $ and $t\ge 1$. We have also completely classified all the extremal word pairs for the models when $t=1$. The techniques developed for the extremal pairs are shown to be useful also to determine codes $C\subseteq \mathbb{Z}_q^n$ which need only a small number of channels to distinguish the transmitted words.  
We have also determined some general bounds on the number of channels in the multiset model for $t\ge 1$. The treatment of the non-multiset case for $t\ge 2$ is more involving (for example, see the conference article \cite{pavanMaxLev} for $t=2$) and is postponed to our later articles. 

There is a wealth of future work in this area. To begin with, the general upper and lower bounds from Section~\ref{sec:generalbds} have some room for improvement. 
Moreover, it would be interesting to know if the techniques in Section~\ref{BoundAndExtremalPairs},  like those used to obtain \eqref{+1Bound} and \eqref{+0Bound}, could be generalized to larger $t$. At least, there seems to be a predictable pattern in the extremal word pairs for larger $t$ as seen in Table~\ref{tab:extremalwords}. Regarding Section~\ref{sec:N_q^m(C;n,1)}, further results would be welcomed that involve different subsets of $\mathbb{Z}_q^n$ and the number of channels required to determine the transmitted word with $t\ge 1$.  

\section*{Appendix
}

In this section, we prove Theorem~\ref{thm:doublerecupperbound}. For this purpose, we first need some lemmas.
 To improve the bound of Corollary~\ref{cor:x1=x1'xn=xn'}, we will study in more detail the result of Proposition~\ref{prop:x1=x1'xn=xn'}. Recalling the assumptions $x_1 = x'_1$ and $x_n = x'_n$, the multiset $M_1(\vec{x}) \cap M_1(\vec{x'})$ therein can be viewed as being equal to $x_1(\{M_1(\vec{x}_{[2, n - 1]}) \oplus \vec{x}_{[2, n]} \oplus \vec{x}_{[1, n - 1]}\} \cap \{M_1(\vec{x'}_{[2, n - 1]}) \oplus \vec{x'}_{[2, n]} \oplus \vec{x'}_{[1, n - 1]}\})x_n$. In the following lemmas, we discuss in detail how the words $\vec{x}_{[2, n]}$, $\vec{x}_{[1, n - 1]}$, $\vec{x'}_{[2, n]}$ and $\vec{x'}_{[1, n - 1]}$ can contribute to the intersection $M_1(\vec{x}) \cap M_1(\vec{x'})$. 

\begin{lemma}
\label{lem:lem1a2a3a4a1}
    Let $\vec{x} = a^k \phi b^\ell$ and $\vec{x'} = a^m \psi b^p$ with $\vec{x} \neq \vec{x'}$, where $k, \ell, m, p \geq 1$, $\phi$ and $\psi$ are words such that neither is their first letter $a$ nor is their last letter $b$, $|\phi| \geq 0$ and $|\psi| \geq 0$. Then $\vec{x}_{[2, n]} \in M_1(\vec{x'}_{[2, n - 1]}) \ominus M_1(\vec{x}_{[2, n - 1]})$ if and only if one of the following conditions is true:
    \begin{enumerate}
        \item $m = k - 1$, $\ell = p - 1$, $m > \ell$ and $\phi = \psi$, or
        \item $m = k$, $\ell = p - 1$ and $m(\phi, M_1(\psi)) > \ell$.
    \end{enumerate}
\end{lemma}
\begin{proof}
    Let $\vec{x} = a^k \phi b^\ell$ and $\vec{x'} = a^m \psi b^p$ with $\vec{x} \neq \vec{x'}$, where $k, \ell, m, p \geq 1$, $\phi$ and $\psi$ are such that neither is their first letter $a$ nor is their last letter $b$, $|\phi| \geq 0$ and $|\psi| \geq 0$. The letters $a$ and $b$ can be equal. Let us assume first that neither $\vec{x}$ nor $\vec{x'}$ is of the form $a^n$. We know that $\vec{x}_{[2, n]} \in M_1(\vec{x'}_{[2, n - 1]}) \ominus M_1(\vec{x}_{[2, n - 1]})$ if and only if with a single insertion, the number of ways of obtaining $\vec{x}_{[2, n]} = a^{k - 1} \phi b^\ell$ from $\vec{x'}_{[2, n - 1]} = a^{m - 1} \psi b^{p - 1}$ is greater than the number of ways of obtaining $a^{k - 1} \phi b^\ell$ from $\vec{x}_{[2, n - 1]} = a^{k - 1} \phi b^{\ell - 1}$. But since the number of ways of obtaining $a^{k - 1} \phi b^\ell$ from $a^{k - 1} \phi b^{\ell - 1}$ using a single insertion is $\ell$, we have the following:
    \begin{equation}
    \label{eq:lem1a2a3a4a1}
    \begin{split}
        \vec{x}_{[2, n]} \in M_1(\vec{x'}_{[2, n - 1]}) \ominus M_1(\vec{x}_{[2, n - 1]})\text{ if and only if } m(a^{k - 1} \phi b^\ell, M_1(a^{m - 1} \psi b^{p - 1})) > \ell.
    \end{split}
    \end{equation}
    We can obtain $a^{k - 1} \phi b^\ell$ from $a^{m - 1} \psi b^{p - 1}$ by a single insertion in the following possible ways:
\begin{enumerate}[(i)]
    \item By adding a letter $a$ before $\psi$: We then must have $m = k - 1$, $\ell = p - 1$ and $\phi = \psi$. There are $m$ ways to do this insertion.
    \item By adding a letter $b$ after $\psi$: We then must have $m = k$, $\ell = p$ and $\phi = \psi$. There are $\ell$ ways to do this insertion.
    \item By adding a letter $c \in \mathbb{Z}_q, c \neq a$ before $\psi$: We then must have $k \leq m $, $\ell = p - 1$ and $\phi = c a^i \psi$ for some $0 \leq i \leq m - 1$. There is only one way to do this insertion.
    \item By adding a letter $c \in \mathbb{Z}_q, c \neq b$ after $\psi$: We then must have $m = k$, $\ell \leq p - 1$ and $\phi = \psi b^i c$ for some $0 \leq i \leq p - 1$. There is only one way to do this insertion.
    \item By adding a letter in $\psi$ (but not as its first or last letter): We then must have $m = k$, $\ell = p - 1$ and $\phi \in M_1(\psi)$. There are $m(\phi, M_1(\psi))$ ways to do this insertion.
\end{enumerate}
Equation (\ref{eq:lem1a2a3a4a1}) requires that the number of ways in which the insertion can be done is greater than $\ell$. Hence, the only possible ways to obtain $a^{k - 1} \phi b^\ell$ from $a^{m - 1} \psi b^{p - 1}$ by a single insertion are represented by (i) and (v), that is, $\vec{x}_{[2, n]} \in M_1(\vec{x'}_{[2, n - 1]}) \ominus M_1(\vec{x}_{[2, n - 1]})$ if and only if either of (i) or (v) holds. This proves the lemma in this case.

Next, suppose that $\vec{x} = a^n$ and $\vec{x'} = a^m \psi a^p$, that is, the case when $\phi = \varepsilon$ and $a = b$. As before, $\vec{x}_{[2, n]} \in M_1(\vec{x'}_{[2, n - 1]}) \ominus M_1(\vec{x}_{[1, n - 1]})$ if and only if $m(a^{n - 1}, M_1(a^{m - 1} \psi a^{p - 1})) > m(a^{n - 1}, M_1(a^{n - 2})) = n - 1$. The word $a^{n - 1}$ can be obtained from $a^{m - 1}\psi a^{p - 1}$ using a single insertion if and only if $|\psi| = 0$. But this implies that $\vec{x} = \vec{x'}$, which is a contradiction. 

Similarly, if $\vec{x} = a^k \phi a^\ell$ and $\vec{x'} = a^n$, that is, if $\psi = \varepsilon$ and $a = b$, then $\vec{x}_{[2, n]} \in M_1(\vec{x'}_{[2, n - 1]}) \ominus M_1(\vec{x}_{[1, n - 1]})$ if and only if $m(a^{k - 1}\phi a^\ell, M_1(a^{n - 2})) > m(a^{k - 1}\phi a^\ell, M_1(a^{k - 1}\phi a^{\ell - 1})) = \ell$. The word $a^{k - 1}\phi a^\ell$ can be obtained from $a^{n - 2}$ using a single insertion if and only if $|\phi| \leq 1$. If $|\phi| = 0$, then as before $\vec{x} = \vec{x'}$, a contradiction. Let $\phi = c \neq a$. Then, $a^{k - 1} c a^\ell$ can be obtained from $a^{n - 2}$ only by inserting $c$ appropriately. Since this insertion can be done in only one way, and $\ell \geq 1$, this is again a contradiction. This completes the proof of the lemma.
\end{proof}

\begin{lemma}
\label{lem:lem1a2a3a4a2}
    Let $\vec{x} = a^k \phi b^\ell$ and $\vec{x'} = a^m \psi b^p$ with $\vec{x} \neq \vec{x'}$, where $k, \ell, m, p \geq 1$, $\phi$ and $\psi$ are such that neither is their first letter $a$ nor is their last letter $b$, $|\phi| \geq 0$ and $|\psi| \geq 0$. Then $\vec{x}_{[1, n - 1]} \in M_1(\vec{x'}_{[2, n - 1]}) \ominus M_1(\vec{x}_{[2, n - 1]})$ if and only if one of the following conditions is true:
    \begin{enumerate}
        \item $m - 1 = k$, $\ell - 1 = p$, $p > k$ and $\phi = \psi$, or
        \item $m - 1 = k$, $\ell = p$ and $m(\phi, M_1(\psi)) > k$.
    \end{enumerate}
\end{lemma}
\begin{proof}
    Let $\vec{x} = a^k \phi b^\ell$ and $\vec{x'} = a^m \psi b^p$ with $\vec{x} \neq \vec{x'}$, where $k, \ell, m, p \geq 1$, $\phi$ and $\psi$ are such that neither is their first letter $a$ nor is their last letter $b$, $|\phi| \geq 0$ and $|\psi| \geq 0$. The letters $a$ and $b$ can be equal. 
    By Corollary~\ref{cor:reverse}, we know that $\vec{x}_{[1, n - 1]} \in M_1(\vec{x'}_{[2, n - 1]}) \ominus M_1(\vec{x}_{[2, n - 1]})$ if and only if $\overleftarrow{\vec{x}_{[1, n - 1]}} \in M_1(\overleftarrow{\vec{x'}_{[2, n - 1]}}) \ominus M_1(\overleftarrow{\vec{x}_{[2, n - 1]}})$. But  $\overleftarrow{\vec{x}_{[1, n - 1]}} =\overleftarrow{\vec{x}}_{[2,n]} = b^{\ell-1}\overleftarrow{\phi} a^k, \overleftarrow{\vec{x'}_{[2, n - 1]}} = \overleftarrow{\vec{x'}}_{[2, n - 1]} = b^{p-1}\overleftarrow{\psi} a^{m-1}$ and $\overleftarrow{\vec{x}_{[2, n - 1]}} = \overleftarrow{\vec{x}}_{[2, n - 1]} =b^{\ell-1}\overleftarrow{\phi} a^{k-1}$. Now, the proof is completed by applying Lemma~\ref{lem:lem1a2a3a4a1}.
\end{proof}

\begin{lemma}
\label{lem:lem1a2a3a4a3}
    Let $\vec{x} = a^k \phi b^\ell$ and $\vec{x'} = a^m \psi b^p$ with $\vec{x} \neq \vec{x'}$, where $k, \ell, m, p \geq 1$, $\phi$ and $\psi$ are such that neither is their first letter $a$ nor is their last letter $b$, $|\phi| \geq 0$ and $|\psi| \geq 0$. Then $\vec{x'}_{[2, n]} \in M_1(\vec{x}_{[2, n - 1]}) \ominus M_1(\vec{x'}_{[2, n - 1]})$ if and only if one of the following conditions is true:
    \begin{enumerate}
        \item $k = m - 1$, $p = \ell - 1$, $k > p$ and $\psi = \phi$, or
        \item $k = m$, $p = \ell - 1$ and $m(\psi, M_1(\phi)) > p$.
    \end{enumerate}
\end{lemma}
\begin{proof}
    This Lemma is symmetric to Lemma~\ref{lem:lem1a2a3a4a1} and can be proved by appropriately replacing the indices related to $\vec{x}$ with the ones corresponding to $\vec{x'}$.
\end{proof}

\begin{lemma}
\label{lem:lem1a2a3a4a4}
    Let $\vec{x} = a^k \phi b^\ell$ and $\vec{x'} = a^m \psi b^p$ with $\vec{x} \neq \vec{x'}$, where $k, \ell, m, p \geq 1$, $\phi$ and $\psi$ are such that neither is their first letter $a$ nor is their last letter $b$, $|\phi| \geq 0$ and $|\psi| \geq 0$. Then $\vec{x'}_{[1, n - 1]} \in M_1(\vec{x}_{[2, n - 1]}) \ominus M_1(\vec{x'}_{[2, n - 1]})$ if and only if one of the following conditions is true:
    \begin{enumerate}
        \item $k - 1 = m$, $p - 1 = \ell$, $\ell > m$ and $\psi = \phi$, or
        \item $k - 1 = m$, $p = \ell$ and $m(\psi, M_1(\phi)) > m$.
    \end{enumerate}
\end{lemma}
\begin{proof}
    This Lemma is symmetric to Lemma~\ref{lem:lem1a2a3a4a2} and can be proved by appropriately replacing the indices related to $\vec{x}$ with the ones corresponding to $\vec{x'}$.
\end{proof}

Recall that when $x_1 = x_1'$ and $x_n = x_n'$, $\vec{x}x_n$ is in $M_1(\vec{x}) \cap M_1(\vec{x'})$ if and only if $\vec{x}_{[2, n]} \in M_1(\vec{x'}_{[2, n - 1]}) \ominus M_1(\vec{x}_{[2, n - 1]})$ or if $\vec{x}_{[2, n]} = \vec{x'}_{[1, n - 1]}$. Similarly, we know that $x_1\vec{x}$ is in $M_1(\vec{x}) \cap M_1(\vec{x'})$ if and only if $\vec{x}_{[1, n - 1]} \in M_1(\vec{x'}_{[2, n - 1]}) \ominus M_1(\vec{x}_{[2, n - 1]})$ or if $\vec{x}_{[1, n - 1]} = \vec{x'}_{[2, n]}$, $\vec{x'}x_n$ is in $M_1(\vec{x}) \cap M_1(\vec{x'})$ if and only if $\vec{x'}_{[2, n]} \in M_1(\vec{x}_{[2, n - 1]}) \ominus M_1(\vec{x'}_{[2, n - 1]})$ or if $\vec{x'}_{[2, n]} = \vec{x}_{[1, n - 1]}$ and $x_1\vec{x'}$ is in $M_1(\vec{x}) \cap M_1(\vec{x'})$ if and only if $\vec{x'}_{[1, n - 1]} \in M_1(\vec{x}_{[2, n - 1]}) \ominus M_1(\vec{x'}_{[2, n - 1]})$ or if $\vec{x'}_{[1, n - 1]} = \vec{x}_{[2, n]}$. The next couple of results determine the properties of $\vec{x}$ and $\vec{x'}$ if $\vec{x}x_n$, $x_1\vec{x}$, $\vec{x'}x_n$ or $x_1\vec{x'}$ are in $M_1(\vec{x}) \cap M_1(\vec{x'})$.

\begin{lemma}
\label{lem:onlyoneoffirstconditions}
Let $\vec{x}, \vec{x'} \in \mathbb{Z}_q^n$ be such that $\vec{x} \neq \vec{x'}$ with $x_1 = x_1', x_n = x_n'$. At most one of the following conditions can be true simultaneously
\begin{enumerate}
    \item $\vec{x}_{[2, n]} \in M_1(\vec{x'}_{[2, n - 1]}) \ominus M_1(\vec{x}_{[2, n - 1]})$,
    \item $\vec{x}_{[1, n - 1]} \in M_1(\vec{x'}_{[2, n - 1]}) \ominus M_1(\vec{x}_{[2, n - 1]})$,
    \item $\vec{x'}_{[2, n]} \in M_1(\vec{x}_{[2, n - 1]}) \ominus M_1(\vec{x'}_{[2, n - 1]})$,
    \item $\vec{x'}_{[1, n - 1]} \in M_1(\vec{x}_{[2, n - 1]}) \ominus M_1(\vec{x'}_{[2, n - 1]})$.
\end{enumerate}
\end{lemma}
\begin{proof}
    Let $\vec{x} = a^k \phi b^\ell$ and $\vec{x'} = a^m \psi b^p$, $\vec{x} \neq \vec{x'}$, where $k, \ell, m, p \geq 1$, $\phi$ and $\psi$ are such that neither is their first letter $a$ nor is their last letter $b$, $|\phi| \geq 0$ and $|\psi| \geq 0$. The letters $a$ and $b$ can be equal. Then, since the conditions of any two lemmas among Lemmas~\ref{lem:lem1a2a3a4a1}, \ref{lem:lem1a2a3a4a2}, \ref{lem:lem1a2a3a4a3} and \ref{lem:lem1a2a3a4a4} contradict each other, it is clear that only at most one of the conditions of this lemma can be true simultaneously.
\end{proof} 

\begin{lemma}
\label{lem:onlytwoofsecondconditions}
    Let $\vec{x}, \vec{x'} \in \mathbb{Z}_q^n$ be such that $\vec{x} \neq \vec{x'}$ with $x_1 = x_1', x_n = x_n'$. Only one of the following conditions can be true simultaneously 
    \begin{itemize}
        \item $\vec{x}_{[2, n]} = \vec{x'}_{[1, n - 1]}$ or 
        \item $\vec{x}_{[1, n - 1]} = \vec{x'}_{[2, n]}.$
    \end{itemize}
\end{lemma}
\begin{proof}
    Suppose that $\vec{x}_{[2, n]} = \vec{x'}_{[1, n - 1]}$ and $\vec{x}_{[1, n - 1]} = \vec{x'}_{[2, n]}$. This means that $x_i = x_{i - 1}'$, for $2 \leq i \leq n$ and $x_j = x_{j + 1}'$, for $1 \leq j \leq n - 1$. Combining these two sets of equalities, we have $x_i = x_{i - 1}' = x_{i - 2}$, for $3 \leq i \leq n$ and $x_{j + 1}' = x_j = x_{j - 1}'$, for $2 \leq j \leq n - 1$. Moreover, since $x_2 = x_1' = x_1$, we get $x_i = x_1$, for $1 \leq i \leq n$ and $\vec{x} = x_1^n$. Similarly, since $x_2' = x_1 = x_1'$, we get $x_j' = x_1'$, for $1 \leq j \leq n$ and $\vec{x'} = x_1'^n$. This gives us that $\vec{x} = \vec{x'}$, which is a contradiction.
\end{proof}

Now, we can finally improve the result in Corollary~\ref{cor:x1=x1'xn=xn'}. The bound in Theorem~\ref{thm:doublerecupperbound} is optimal in the sense that it can be attained (see  Theorem~\ref{thm:lowerboundt=1odd}).

\begin{proof} (Proof of Theorem~\ref{thm:doublerecupperbound}):
    Observe that from Proposition~\ref{prop:x1=x1'xn=xn'}, whether $\vec{x}_{[2, n]}, \vec{x}_{[1, n - 1]}, \vec{x'}_{[2, n]}$ and $\vec{x'}_{[1, n - 1]}$ are in $M_1(\vec{x}) \cap M_1(\vec{x'})$, respectively depend on whether the following conditions hold:

\begin{enumerate}
    \item \begin{enumerate}
            \item $\vec{x}_{[2, n]} \in M_1(\vec{x'}_{[2, n - 1]}) \ominus M_1(\vec{x}_{[2, n - 1]})$, or
            \item $\vec{x}_{[2, n]} = \vec{x'}_{[1, n - 1]}.$
        \end{enumerate}
    \item \begin{enumerate}
            \item $\vec{x}_{[1, n - 1]} \in M_1(\vec{x'}_{[2, n - 1]}) \ominus M_1(\vec{x}_{[2, n - 1]})$, or
            \item $\vec{x}_{[1, n - 1]} = \vec{x'}_{[2, n]}.$
        \end{enumerate}
    \item \begin{enumerate}
            \item $\vec{x'}_{[2, n]} \in M_1(\vec{x}_{[2, n - 1]}) \ominus M_1(\vec{x'}_{[2, n - 1]})$, or
            \item $\vec{x'}_{[2, n]} = \vec{x}_{[1, n - 1]}.$
        \end{enumerate}
    \item \begin{enumerate}
            \item $\vec{x'}_{[1, n - 1]} \in M_1(\vec{x}_{[2, n - 1]}) \ominus M_1(\vec{x'}_{[2, n - 1]})$, or
            \item $\vec{x'}_{[1, n - 1]} = \vec{x}_{[2, n]}.$
        \end{enumerate}
\end{enumerate}
From Lemmas~\ref{lem:onlyoneoffirstconditions} and~\ref{lem:onlytwoofsecondconditions}, the only possible case in which $|M_1(\vec{x}) \cap M_1(\vec{x'})| > |M_1(\vec{x}_{[2, n - 1]}) \cap M_1(\vec{x'}_{[2, n - 1]})| + 1$ is if without loss of generality, the conditions 1)a), 2)b) and 3)b) are true. That is, if $\vec{x}_{[2, n]} \in M_1(\vec{x'}_{[2, n - 1]}) \ominus M_1(\vec{x}_{[2, n - 1]})$ and $\vec{x}_{[1, n - 1]} = \vec{x'}_{[2, n]}$. 

Let $\vec{x} = a^k \phi b^\ell$,$\vec{x'} = a^m \psi b^p$, $\vec{x} \neq \vec{x'}$, where $k, \ell, m, p \geq 1$, $\phi$ and $\psi$ are such that neither is their first letter $a$ nor is their last letter $b$, $|\phi| \geq 0$ and $|\psi| \geq 0$. The letters $a$ and $b$ can be equal. Then we know from Lemma~\ref{lem:lem1a2a3a4a1} that $\vec{x}_{[2, n]} \in M_1(\vec{x'}_{[2, n - 1]}) \ominus M_1(\vec{x}_{[2, n - 1]})$ if and only if one of the following conditions holds
\begin{enumerate}[(i)]
    \item $m = k - 1$, $\ell = p - 1$, $m > \ell$ and $\phi = \psi$, or
    \item $m = k$, $\ell = p - 1$ and $\phi \in M_1(\psi)$.
\end{enumerate}
Moreover, since $\vec{x}_{[1, n - 1]} = \vec{x'}_{[2, n]}$, that is $a^k \phi b^{\ell - 1} = a^{m - 1} \psi b^p$, we have $ k = m - 1, p = \ell - 1$ and $\phi = \psi$. However, this contradicts both the conditions for $\vec{x}_{[2, n]} \in M_1(\vec{x'}_{[2, n - 1]}) \ominus M_1(\vec{x}_{[2, n - 1]})$. This completes the proof.
\end{proof}

After proving Theorem~\ref{thm:doublerecupperbound}, we consider in the following remark how Lemmas~\ref{lem:lem1a2a3a4a1}--\ref{lem:lem1a2a3a4a4} and~\ref{lem:onlytwoofsecondconditions} can be used for determining whether~\eqref{+1Bound} or \eqref{+0Bound} holds with equality, that is, whether $|M_1(\vec{x}) \cap M_1(\vec{x'})|$ is equal to $ |M_1(\vec{x}_{[2, n - 1]}) \cap M_1(\vec{x'}_{[2, n - 1]})| + 1$ or to $|M_1(\vec{x}_{[2, n - 1]}) \cap M_1(\vec{x'}_{[2, n - 1]})|$.
\begin{remark} \label{RmkCharacterization}
    Lemma~\ref{lem:lem1a2a3a4a1} gives a characterization for $\vec{x}_{[2, n]}$ to belong to $M_1(\vec{x'}_{[2, n - 1]}) \ominus M_1(\vec{x}_{[2, n - 1]})$, that is, $\vec{x}_{[2, n]}$ to contribute to the intersection $M_1(\vec{x}) \cap M_1(\vec{x'})$ due to it having greater multiplicity in $M_1(\vec{x'}_{[2, n - 1]})$ than in $M_1(\vec{x}_{[2, n - 1]})$. Similarly, Lemma~\ref{lem:lem1a2a3a4a2} considers the contribution of $\vec{x}_{[1, n-1]}$ to the intersection. Furthermore, Lemmas~\ref{lem:lem1a2a3a4a3} and~\ref{lem:lem1a2a3a4a4} study (respectively) when $\vec{x'}_{[2, n]}$ and $\vec{x'}_{[1, n-1]}$ belong to $M_1(\vec{x}_{[2, n - 1]}) \ominus M_1(\vec{x'}_{[2, n - 1]})$, that is, when they contribute to the intersection $M_1(\vec{x}) \cap M_1(\vec{x'})$. Finally, in Lemma~\ref{lem:onlytwoofsecondconditions}, we study how $\vec{x}_{[2, n]}$, $\vec{x'}_{[2, n]}$, $\vec{x}_{[1, n-1]}$ and $\vec{x'}_{[1, n-1]}$ may contribute to the intersection if the conditions of Lemmas~\ref{lem:lem1a2a3a4a1}--\ref{lem:lem1a2a3a4a4} are not satisfied. In conclusion, the previous lemmas provide a way to characterize whether $|M_1(\vec{x}) \cap M_1(\vec{x'})|$ is equal to $ |M_1(\vec{x}_{[2, n - 1]}) \cap M_1(\vec{x'}_{[2, n - 1]})| + 1$ or to $|M_1(\vec{x}_{[2, n - 1]}) \cap M_1(\vec{x'}_{[2, n - 1]})|$.
\end{remark}

In the following example, we demonstrate how the previous results can be used to show that the bound in \eqref{+0Bound} is tight. Notice that the code $C=\{0010,0110\}\subseteq \mathbb{Z}_2^4$ in the remark is the one mentioned in Section~\ref{sec:N_q^m(C;n,1)}.

\begin{example}\label{+0Conditions}
    Let  $C=\{0010,0110\}\subseteq \mathbb{Z}_2^4$. It is easy to check that $M_1(0010)\cap M_1(0110)=\{00110, 00110, 01010\}$ and thus, $N_2^m(C;4,1)=3$. Let $C'\subseteq\mathbb{Z}_2^4$ be the code with eight codewords described in \eqref{subsetconstr} of Theorem~\ref{constr}. In order to show that $N_2^m(C';6,1)\le 3$ it is enough to consider codeword pairs of $C'$ that have the same first symbols and the same last symbols (since, for all other pairs, $|M_1(\vec{x})\cap M_1(\vec{x'})|\le 2$). There are four such pairs $\{a0010b,a0110b\}$, where $ab\in \mathbb{Z}_2^2$. It suffices to show that for such codeword pairs the bound \eqref{+0Bound} holds with equality (instead of \eqref{+1Bound}). We consider here only the case $a=b=0$, i.e., the codeword pair $\vec{x}=000100$ and $\vec{x'}=001100$ (the other pairs work analogously). 
    For these words, we have (using the notation of Lemmas~\ref{lem:lem1a2a3a4a1}--\ref{lem:lem1a2a3a4a4})
    $k=3$, $\ell=2$, $m=2$, $p=2$, $\phi=1$ and $\psi=11$.
    The conditions 1) and 2) of Lemma~\ref{lem:lem1a2a3a4a1} do not hold since $\ell\neq p-1$. The conditions 1) and 2) of  Lemma~\ref{lem:lem1a2a3a4a2} (resp. Lemma~\ref{lem:lem1a2a3a4a3})  are not satisfied due to $m-1\neq k$ (resp. $p\neq \ell-1$). Let us then consider Lemma~\ref{lem:lem1a2a3a4a4}. The condition 1) does not hold since $p-1\neq \ell$. For the condition 2) we do have $k-1=m$ and $p=\ell$ but $m(\psi,M_1(\phi))=m(11,M_1(1))=2>m=2$ does not hold.
    Therefore, the only thing we need to check for the equality holding in  \eqref{+0Bound} is that neither of the conditions of Lemma~\ref{lem:onlytwoofsecondconditions} are satisfied. This is indeed the case as  $\vec{x}_{[2, n]}=00100 \neq 00110=\vec{x'}_{[1, n - 1]}$ and
         $\vec{x}_{[1, n - 1]}=00010 \neq 01100= \vec{x'}_{[2, n]}.$    
\end{example}

Notice that an alternative way to prove that $N_2^m(C';6,1)\le 3$ can be found using Theorem~\ref{thm:extremalmulteven}; indeed, by Corollary~\ref{cor:exactboundt=1}, we have $N_2^m(C';6,1)\le N_2^m(6,1) = 4$, and $C'$ do not contain any extremal pairs indicated by Theorem~\ref{thm:extremalmulteven}.

\bibliographystyle{IEEEtran}
\bibliography{main}

\end{document}